%% file: bjps-sample.tex
\documentclass[
  journal=bjps,
]{cup-journal}

\usepackage{booktabs,microtype,siunitx,tabularx,relsize,appendix}

\newcommand{\SECO}{\mathrm{SECO}}

\title{Identifying regions of concomitant compound precipitation and wind speed extremes over Europe}
\shorttitle{Identifying concomitant compound precipitation and wind speed extremes}

\author{Alexis Boulin}
\affiliation{Université Côte d’Azur, CNRS, LJAD, France}
\affiliation{Univ Montpellier, CNRS, Montpellier, France}
\alsoaffiliation{Inria, Lemon}
\email{aboulin@unice.fr}
\author{Elena Di Bernardino}
\affiliation{Université Côte d’Azur, CNRS, LJAD, France}
\author{Thomas Lalo\"{e}}
\affiliation{Université Côte d’Azur, CNRS, LJAD, France}
\author{Gwladys Toulemonde}
\affiliation{Univ Montpellier, CNRS, Montpellier, France}
\alsoaffiliation{Inria, Lemon}

\keywords{Compound events, Extremes, High dimensional models, Spatial clustering, Variable clustering.}
\begin{document}

\input{main}

\printbibliography

\appendix

\input{appendix_main}

\end{document}

%% file: main.tex
\begin{abstract}
    The task of simplifying the complex spatio-temporal variables associated with climate modeling is of utmost importance and comes with significant challenges. In this research, our primary objective is to tailor clustering techniques to handle compound extreme events within gridded climate data across Europe. Specifically, we intend to identify subregions that display asymptotic independence concerning compound precipitation and wind speed extremes. To achieve this, we utilise daily precipitation sums and daily maximum wind speed data derived from the ERA5 reanalysis dataset spanning from 1979 to 2022. Our approach hinges on a tuning parameter and the application of a divergence measure to spotlight disparities in extremal dependence structures without relying on specific parametric assumptions. We propose a data-driven approach to determine the tuning parameter. This enables us to generate clusters that are spatially concentrated, which can provide more insightful information about the regional distribution of compound precipitation and wind speed extremes. In the process, we aim to elucidate the respective roles of extreme precipitation and wind speed in the resulting clusters. The proposed method is able to extract valuable information about extreme compound events while also significantly reducing the size of the dataset within reasonable computational timeframes.
\end{abstract}
  
\section*{Introduction}
The occurrence of extreme weather events is often exacerbated by the convergence of distinct geographic factors and concurrent weather patterns, resulting in profound disruptions and extensive damage to society. Catastrophic climate phenomena such as floods, wildfires, and heatwaves frequently manifest due to the simultaneous intensification of multiple interacting processes. When these various processes coalesce to yield a substantial impact, it is referred to as a compound event. Among the primary manifestations of extreme weather, heavy precipitation and robust surface winds hold central positions, exerting adverse effects on both the natural world and human society. Extratropical cyclones, along with their associated wind patterns and storm surges, contribute significantly to economic and insured losses resulting from natural calamities in Europe. Furthermore, they disrupt transportation, trade, and energy supply systems, often leading to human casualties (refer to, for instance, \cite{schwierz2010modelling, pinto2012loss}). To mitigate these impacts, it is important to better understand the dependence structure of extreme weather events. However, modeling such complex scenarios, where multiple rare events occur simultaneously, can be incredibly challenging, especially with high-dimensional climate datasets that exhibit heavy tails. 

\begin{figure}[!htb]
    \centering
    \includegraphics[scale=.5]{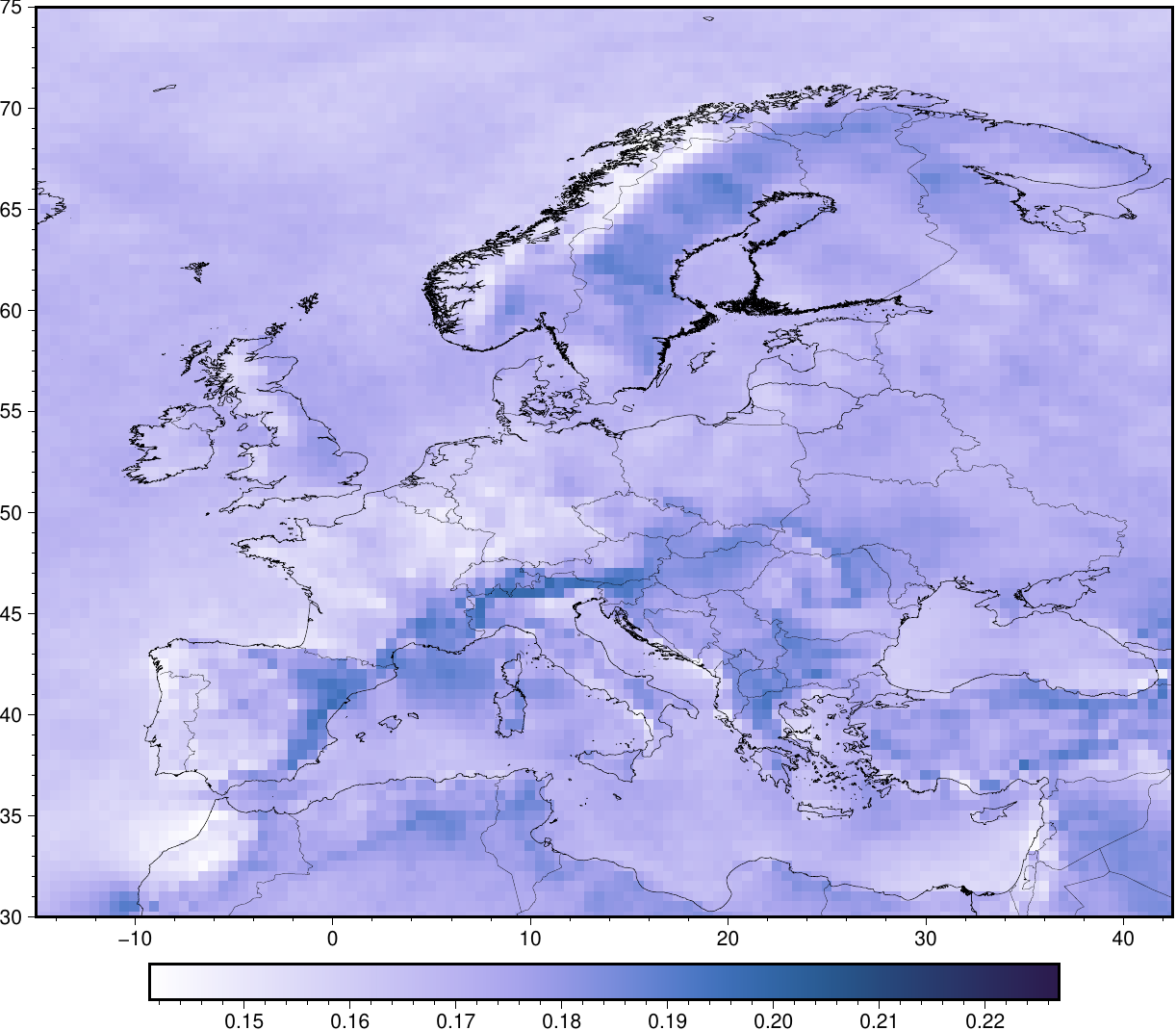}
    \caption{Proportion of the total precipitation or wind speed in the ERA5 dataset that exceed their respective $0.9$th quantiles.}
    \label{subfig:introduction_sum}
\end{figure}

\paragraph{Considered ERA5 dataset} Large ensemble simulations present a unique opportunity for gaining deeper insights into the spatial regionalisation of compound precipitation and wind extremes. This is primarily because these simulations provide a more accurate representation of local-scale processes compared to global ensembles and do so without being hindered by data limitations. However, such simulations need to be interpreted with care as it is often largely unknown how well the employed models represent observed compound events (\cite{zscheischler2021evaluating}), and differences might be large between models. In our endeavor to regionalize compound precipitation and wind speed, we turn to the ERA5 dataset (\cite{CDS}). This dataset allows us to investigate the correlation between daily cumulative precipitation and daily peak wind speeds throughout the extended winter between season, spanning from November to March, across Europe, for the period between 1979 and 2022. ERA5 offers a comprehensive record of atmospheric conditions, land surface characteristics, and ocean wave patterns, spanning from 1950 to the present day. It is worth noting that ERA5 supersedes the previous ERA-interim reanalysis, which began in 1979 and was initiated in 2006.

ERA5 has benefited from significant advancements in model physics, core dynamics, and data assimilation techniques developed over the past decade. Produced using 4D-Var data assimilation technology within model cycle 41r2 (Cy41r2), ERA5 incorporates improved parameterization schemes (\cite{herbaschera5}). The dataset is available at a spatial resolution of $0.25^\circ$ (approximately $27$-$28$ km) on a regular grid. Our specific focus lies within the region defined by $[-15^\circ E, 42.5^\circ E] \times [30^\circ N, 75^\circ N]$, which encompasses Europe. We remap the original hourly data to a regularly spaced grid with a $0.5^\circ$ spatial resolution, allowing us to compute daily precipitation totals and daily maximum wind speeds. We selected the $0.5^\circ$ spatial resolution due to its ability to facilitate calculations within a reasonable timeframe while maintaining manageable storage requirements. The need for remapping can be circumvented with more extensive computing resources. The resulting dataset comprises $6655$ daily precipitation totals and maximum wind speed measurements, covering $91 \times 116$ grid cells with the chosen spatial resolution, totaling $10556$ grid cells for clustering. To illustrate, Fig. \ref{subfig:introduction_sum} provides a visualization of the proportion of grid cells where either wind speed or precipitation exceed a significant threshold. As observed in both panels, there is a noticeable spatial variation in these proportions. In the following, we introduce some notations to describe the spatio-temporal process under consideration.

Consider a spatio-temporal random field denoted as $(\textbf{Z}_n^{(s)}, s \in D \subset \mathbb{R}^2, n \in \mathbb{N})$. Here, $\textbf{Z}_n^{(s)} = (Z_n^{(s,1)}, Z_n^{(s,2)})$ represents the vector of daily total precipitation and wind speed maxima at location $s$ on day $n$. We assume that $\textbf{Z}_n^{(s)}$ is identically distributed over $n$ for each location $s$ in the domain $D$. Now, let's suppose that we have observations available at $d$ spatial locations for each time $n$. We can represent these observations as $\textbf{Z}_n = (\textbf{Z}_n^{(1)},\dots, \textbf{Z}_n^{(d)})$, where $\textbf{Z}_n$ is random vector with stationary law $\textbf{Z} = (\textbf{Z}^{(1)},\dots,\textbf{Z}^{(d)})$. Each $\textbf{Z}^{(j)}$ is a random vector with two dimensions, corresponding to precipitation and wind speed. In the given dataset, each variable represents a location on a grid pixel scale and is characterised by multiple features that can exhibit extreme behavior when considered together. This implies that clustering methods designed to analyze the dependence structure should consider spatial extremal dependence across different locations. In the following discussion, we will review some studies that aim to understand the dependence structure of potentially high-dimensional random vectors.

\paragraph{Related literature} In the field of high-dimensional extremes, researchers have made significant contributions to address the challenges associated with vectors that have univariate random margins, with a focus on unsupervised techniques such as support identification (see, e.g., \cite{goix2017sparse, meyer2021sparse}), Principal Component Analysis of the angular measure of extremes (\cite{cooley2019decompositions, 10.1214/21-EJS1803}), graphical models for extremes (\cite{hitz2016one, engelke2020graphical, asenova2021inference}) and clustering methods (\cite{10.1214/20-EJS1689, boulin2023high}). These methods can help identify hidden spatial patterns and sub-regions where variables are dependent on their extremes, which is crucial for regionalisation tasks.

Over the years, a number of clustering approaches have been suggested, with a focus on extremes, based on a comparison between univariate distributions. For instance, \cite{bernard2013clustering} analysed weekly maxima of precipitation in France and developed a clustering algorithm on a proper distance, the madogram, justified by Extreme Value Theory (EVT). The same approach was used in \cite{bador2015spatial} to evaluate the bias of climate model simulations of temperature maxima over Europe. In \cite{durante2015clustering}, a four-step clustering procedure was presented that considered a pairwise conditional Spearman's correlation coefficient, extracted from daily-log-returns of the adjusted stock price, as a measure of tail dependence. \cite{pappada2018clustering} investigated spatial sub-regions (clusters) of flood risk behavior in the Po river basin in Italy using a copula-based Agglomerative Hierarchical Clustering. In the paper by \cite{maume2023detecting}, they introduce a modified spectral clustering algorithm designed for analyzing spatial extreme events. This algorithm combines spectral clustering with the concept of extremal concurrence probability, as proposed by \cite{dombry2018probabilities}. The goal of this approach is to determine whether a max-stable process exhibits a stationary dependence structure or not. Recently, \cite{boulin2023high} proposed a class of models, the Asymptotic Independent block (AI block) models, for variable clustering, which defines population-level clusters based on the independence of extremes between clusters. They exhibited an algorithm that compares the extremal dependence of univariate distributions at different locations and showed that it recovers the thinnest partition such that extremes between groups of random variables are mutually independent with high probability.

While regional analysis of univariate climate extreme events is a well-studied area of research, multivariate compound extreme events at larger scales have received less attention. Although the widely known Kullback-Leibler divergence has been adapted for use in the context of compound extreme events, it has been primarily employed to cluster data based on their bivariate extreme behavior (as demonstrated in \cite{vignotto2021clustering}) and to analyze compound weather and climate events (as discussed in \cite{zscheischler2021evaluating}). However, this metric primarily summarises the differences in distribution between two sets of random variables when at least one of their components is extreme, and it does not quantify deviations from asymptotic independence. recognising sub-regions characterised by concurrent extreme precipitation and wind speed events is essential for improving extreme event modeling. This is particularly evident in works like \cite{10.1214/19-AOS1836} and \cite{engelke2020graphical}, which rely on the assumption of asymptotic dependence in the data. Such insights are crucial for the development of strategies to mitigate the impacts of these extreme events.

\paragraph{Proposed methodology.} In this paper, our objective is to expand upon the AI block model as introduced in \cite{boulin2023high} to tackle the challenges posed by this environmental dataset. We depart from the assumption that clusters of pixels are mutually independent univariate time series, with regard to their extremes. Instead, we shift to a framework where a collection of univariate time series is recorded for each pixels, with a particular emphasis on their extreme behavior. To tackle this intrinsic problem, we introduce the concept of \emph{constrained} AI block model, compelling pixels represent a collection of univariate time series. This concept comes into play in our environmental problematic which concerns phenomena like precipitation and wind speed extremes recorded at a specific geographic locations represented as pixels within a large ensemble dataset.

Our objective is the following: cluster a number of $d=10556$ pixels across Europe based on their asymptotic independence on compound precipitation and wind speed extremes where data are relatively scarse, i.e., the sample size $n = 6655$. To efficiently implement a fast algorithm designed for this model-based approach in such a high-dimensional setting, we employ a divergence measure that highlights the differences in extremal dependence structures for asymptotically dependent and independent random vectors. Noticeably, this divergence measure adheres to several axioms that makes it a valid measure of dependence (\cite{dekeysercopula2023,de2023parametric}) and also a coherent measure (\cite{scarsini1984measures}). Furthermore, this divergence is linked to a well-known quantity in Extreme Value Theory and can be consistently estimated without the need of parametric assumptions. This consistent estimation is possible under the condition of weak mixing conditions to stay within the scope of our application where departures from the independence assumption are strongly suspected. 

The algorithm requires the specification of a tuning parameter, and we suggest an approach based on data to determine its value. When applied to our environmental dataset, this clustering procedure is efficient and produces clusters that are spatially concentrated, which is a pattern commonly observed in spatial processes. Furthermore, we leverage the interpretability of classical AI block models to gain insights into the role of precipitation and wind speed extremes in the compound partition. In other words, we use this approach to understand how the clustering is influenced by both wind speed and precipitation. To further analyze the results, we make use of a straightforward modification of our dissimilarity measure. This modification allows us to comment on the different clusterings obtained through various algorithms and provide valuable insights into our proposed methodology.

\paragraph{Structure of the paper.} The remainder of the paper is structured as follows. Section \ref{sec:clt_alg} of the paper provides an introduction to extreme value theory, followed by a detailed explanation of the method used to cluster compound extremes. The  proposed clustering process relies an extremal dependence measure for random vectors,  recently introduced in  \cite{boulin2023high}: the Sum of Extremal COefficient ($\SECO$).  In Section \ref{sec:results} we present the  results obtained from applying the clustering algorithm for considered precipitation and wind speed extremes over  Europe.  Conclusions of our study and possible improvement of our method are explored in Section \ref{sec:conclusion}. 
{In the external supplementary materials, we postponed some material  about dependence measures for random vectors (see Section \ref{sec:axioms}).  Technical arguments are given in \ref{sec:coherent_measure}, \ref{sec:incompleteness} and \ref{proof:prop_consistency}. Finally additional figures are provided in \ref{sec:sup_fig}}. 

\section{A clustering algorithm for compound extreme events}
\label{sec:clt_alg}

\subsection{A measure for evaluating dependence between compound extremes}
\label{sec:measure_dependence}

We consider a high-dimensional random vector $\textbf{Z} = (\textbf{Z}^{(1)},\dots,\textbf{Z}^{(d)})$ with law $F$ having $d$ marginal random vectors $\textbf{Z}^{(j)} = (Z^{(j,1)},\dots, Z^{(j,p_j)})$ for $j = 1, \dots, d$. To accommodate vectors of different sizes which will be useful later, we introduce a different notation from the one given in the introduction. Each $\textbf{Z}^{(j)}$ contains $p_j$ marginal univariate random variables $Z^{(j,\ell)}$ for $\ell = 1,\dots,p_j$. In this framework $\textbf{Z}$ has $q = p_1 + \dots + p_d$ marginal univariate random variables.

We call for convenience a function $u$ on $\mathbb{R}$ a normalising function if $u$ is non-decreasing, right continuous, and $u(x) \rightarrow \pm \infty$ as $x \rightarrow \pm \infty$. For a stationary sequence $(\textbf{Z}_n, n \in \mathbb{N})$ of $\textbf{Z}$, we say that the distribution $F$ belongs to the max-domain of attraction of the Extreme Value Distribution (EVD) $H$ if the following convergence result holds for properly normalised maxima:
\begin{equation}
\label{eq:mda}
\mathbb{P}\left\{ \bigvee_{i=1}^n \textbf{Z}_i \leq \textbf{u}_n(\textbf{x}) \right\} \underset{n \rightarrow \infty}{\longrightarrow} H(\textbf{x}), \quad \textbf{x} \in \mathbb{R}^q,
\end{equation}
where $\textbf{u}_n (\textbf{x})=(u_n^{(1,1)}(x^{(1,1)}), \dots, u_n^{(1,p_1)}(x^{(1,1)}),\dots,u_n^{(d,p_d)}(x^{(d,p_d)}))$ is a $q$-dimensional vector of normalising functions. The margins $H^{(1,1)},\dots,H^{(d,p_d)}$ of $H$ must be univariate extreme value distributions and the dependence structure of $H$ is determined by the relation
\begin{equation}
\label{eq:evd}
-\ln H(\textbf{x}) = L\left(-\ln H^{(1,1)}(x^{(1,1)}),\dots, -\ln H^{(d,p_d)}(x^{(d,p_d)})\right),
\end{equation}
for all points $\textbf{x}$ such that $H^{(j,\ell)}(x^{(j,\ell)}) > 0$ for all $j=1,\dots,d$, $\ell = 1,\dots,p_j$. The convergence result in \eqref{eq:mda} with the relation in \eqref{eq:evd} holds under mild assumptions on the dependence between $\textbf{Z}_{1},\dots,\textbf{Z}_{n}$, making the study of time series relevant within the framework of EVT. The stable tail dependence function $L : [0, \infty)^q \rightarrow [0, \infty)$ can be retrieved from the distribution function $F$ via
\begin{equation}\label{eq:stable_from_f}
    L(\textbf{x}) = \underset{t \rightarrow 0}{\lim} \, t^{-1} \mathbb{P}\left\{ F^{(1,1)}(Z^{(1,1)}) > 1 - t x^{(1,1)} \textrm{ or } \dots \textrm{ or } F^{(d,p_d)}(Z^{(d,p_d)}) > 1 - t x^{(d,p_d)}\right\}.
\end{equation}

We assume that the random vector $\textbf{Z}$ is in the max-domain of attraction of an EVD, denoted as $H$. Moreover, we aim that the extremal dependence can be modelled using an AI block model (see \cite{boulin2023high}) where the definition is recalled below.   
\begin{definition}[Asymptotic  Independent block model]
    \label{def:ai_bm}
    Let $(\normalfont{\textbf{Z}_n}, n \in \mathbb{N})$ be a $q$-variate stationary random sequence with law $F$ in the max domain of attraction of $H$. The random sequence $(\normalfont{\textbf{Z}_n}, n \in \mathbb{N})$ is said to follow an AI block model if there exists a partition $O = \{O_g\}_{g=1}^G$ of $\{1,\dots,q\}$ with $|O_g| = d_g$ and marginal extreme value distributions $H^{(O_g)} : \mathbb{R}^{d_g} \rightarrow [0,1]$ such that $\normalfont{H = \Pi_{g=1}^G H^{(O_g)}}$.
\end{definition}

The constrained AI block model requires improvements to the methods proposed in \cite{boulin2023high}, which uses extremal correlation to detect asymptotic independence between random variables, to correctly identify the hidden partition. Asymptotic independence is a concept that describes the relationship between extremes of two random variables, denoted $Z^{(a)}$ and $Z^{(b)}$. Each variable has its own cumulative distribution function, denoted $F^{(a)}$ and $F^{(b)}$, respectively. The extremal correlation, denoted as $\chi(a,b)$, between these two random variables is formally stated by 
\begin{equation*}
    \chi(a,b) = \underset{t \rightarrow 0}{\lim} \, \mathbb{P}\left\{ F^{(a)}(Z^{(a)}) >1-t | F^{(b)}(Z^{(b)}) > 1-t \right\}.
\end{equation*}
The extremal correlation represents the probability of one variable being extreme given that the other is also extreme. If the extremal correlation coefficient $\chi(a,b)$ is in the range of $(0,1]$, then $Z^{(a)}$ and $Z^{(b)}$ are said to be asymptotically dependent. Otherwise, if $\chi(a,b) = 0$, the variables are asymptotically independent. For instance, the well-known bivariate Gaussian distribution with correlation coefficient $\rho \in [-1,1)$ satisfies $\chi(a,b) = 0$.

An extension beyond the bivariate case involves examining two groups of random variables. In \cite{boulin2023high}, a new metric called Sum of Extremal COefficient ($\SECO$) was introduced. To better understand the definition of the metric, the necessary notations for extremal coefficients of an extremal random vector with possibly different sizes are presented below:
\begin{align}
    &\theta(1,\dots,d) = \underset{t \rightarrow 0}{\lim} \, t^{-1} \mathbb{P}\left\{ \underset{j =1,\dots,d}{\max} \,\underset{\ell = 1,\dots,p_j}{\max} F^{(j,\ell)}(Z^{(j,\ell)}) > 1-t\right\}\\
    &\theta(j) = \underset{t \rightarrow 0}{\lim} \, t^{-1} \mathbb{P}\left\{ \underset{\ell = 1,\dots,p_j}{\max} F^{(j,\ell)}(Z^{(j,\ell)}) > 1-t\right\}, \quad j = 1,\dots,d.
\end{align}
Then, $\theta(j)$ corresponds to the extremal coefficient associated to the $j$th marginal random vector $\textbf{X}^{(j)}$. The $\SECO$ metric is defined as the difference between the sum of the extremal coefficients of the $d$ marginal random vectors $\textbf{Z}^{(j)}$ and the extremal coefficient of the entire vector, that is, formally stated
\begin{equation}
    \label{eq:seco}
    \SECO(\textbf{Z}^{(1)},\dots,\textbf{Z}^{(d)}) = \sum_{j=1}^d \theta(j) - \theta(1,\dots,d).
\end{equation}
The $\SECO$ metric is always positive and quantifies the deviation from asymptotic independence between the $d$ groups of variables. Indeed, \cite[Proposition 8 (iii)]{boulin2023high}, \cite{FERREIRA2011586} showed that the $\SECO$ metric is equal to zero if and only if the $d$ groups of variables are independent extreme value random vectors. Moreover, the bivariate $\SECO$ between $\textbf{Z}^{(j)}$ and $\textbf{Z}^{(k)}$ simplifies to the extremal correlation when these are random variables. Recently, \cite{de2023parametric} developed an axiomatic framework to quantify dependence between multiple groups of random variables of possibly different sizes. For self-consistency, we recall them in \ref{sec:axioms} and we show that the $\SECO$ metric defined in \eqref{eq:seco} satisfies most of the stated axioms (see Lemma \ref{prop:dep_meas} in \ref{sec:axioms}). Also, we are interested in determining whether $\SECO$ remains coherent (as defined in \cite{durante2015principles, scarsini1984measures}) when comparing two random vectors with the same dimension, which occurs when $p_a = p_b$. In \ref{sec:coherent_measure} of the supplementary materials, we examine the coherence of the $\SECO$ in nested extreme value copulae, which were introduced in \cite{hofert2018hierarchical}. A further objective is to investigate how SECO behaves in specific nested models.

\subsection{Clustering for compound extremes}
\label{subsec:clt_comp_ext}
In this section, we introduce a modified version of the ECO algorithm as presented in \cite{boulin2023high}, which is capable of clustering compound extremes.

To introduce flexibility in AI block models and bring notations into our application, we consider $\textbf{Z}^{(j)}_i$, $i = 1,\dots,n$, $j = 1, \dots, d$, which are extracted from $d$ pixels. We assume that each observation is distributed according to $F$ in the max-domain of attraction an \emph{constrained} AI block model as defined in \ref{sec:measure_dependence}. Any deviation from asymptotic independence between two pixels can then be measured by the empirical counterpart version of the $\SECO$ in \eqref{eq:seco}. The empirical counterpart of the $\SECO$ between two pixels $a$ and $b$ is defined as:
\begin{equation}
\label{eq:emp_seco}
\widehat{\SECO}(\textbf{Z}^{(a)},\textbf{Z}^{(b)}) = \hat{\theta}(a) + \hat{\theta}(b) - \hat{\theta}(a,b),
\end{equation}
In this case study, we consider $d = 10556$ pixels. For each pixel $j = 1, \dots, d$, we define a bivariate random vector $\textbf{Z}^{(j)} = (Z^{(j,1)}, Z^{(j,2)})$, where, as a convention in this paper, $Z^{(j,1)}$ and $Z^{(j,2)}$ represent the stationary distributions of daily total precipitation and wind speed maxima at location $j$, respectively. Hence, to stick in this context, $\hat{\theta}(a,b)$ is the empirical counterpart of the extremal coefficient for the joint vector $(\textbf{Z}^{(a)},\textbf{Z}^{(b)})$, and $\hat{\theta}(j)$ is the empirical counterpart of the extremal coefficient for the random vector $\textbf{Z}^{(j)}$, $j \in \{a,b\}$, i.e.,
\begin{align}
    &\hat{\theta}(a,b) = \frac{1}{k} \sum_{i=1}^n \mathds{1}_{ \{ R_i^{(a,1)} > n - k + 0.5 \textrm{ or } R_i^{(a,2)} > n - k + 0.5 \textrm{ or } R_i^{(b,1)} > n - k + 0.5 \textrm{ or } R_i^{(b,2)} > n - k +0.5\} } \label{eq:emp_eco} \\
    &\hat{\theta}(j) = \frac{1}{k} \sum_{i=1}^n \mathds{1}_{ \{ R_i^{(j,1)} > n - k + 0.5 \textrm{ or } R_i^{(j,2)} > n - k + 0.5 \} }, \quad j = a,b, \label{eq:emp_seco_sole}
\end{align}
where $R_i^{(j,\ell)}$ denote the rank of $Z_i^{(j,\ell)}$ among $Z_1^{(j,\ell)}, \dots, Z_n^{(j,\ell)}$, $j = a,b$, $\ell = 1,2$. The notation described above can be easily extended to pixels of varying sizes. However, for the sake of clarity in notation and to maintain consistency with our application, we focus on the scenario where pixels consist of bivariate time series. The statistic in Equation \eqref{eq:emp_seco} gauges the strength of dependence between two pixels by counting how often one of the pixels is extreme in terms of either precipitation or wind speed, while adjusted for the number of times a component is extreme in at least one of the pixels.

Classical non-parametric estimators of the extremal coefficient, as demonstrated in equation \eqref{eq:emp_eco}, is only relevant when the number of variables is less than the number of observations, meaning $d \leq n$. Indeed, in \ref{sec:incompleteness} (see Lemma \ref{lem:inequality_EKS} and Lemma \ref{lem:inequality_mado}), we exhibit bounds which show limitations that classical estimators of the extremal dependence structure face, preventing them from encompassing the full range of extremal dependence structures. In high-dimensional scenarios where the number of variables exceeds the number of observations ($d > n$), classical estimators may fail to identify asymptotic independence in extremes. This is particularly relevant for estimators such as Equation \eqref{eq:emp_seco_sole}. This phenomenon can also be observed in other estimators, such as the madogram, and other rank-based estimators of the tail dependence structure may suffer from the same issue. Therefore, caution must be exercised when dealing with high-dimensional data, particularly by taking a lower number of extremes (determined by the parameter $k$), as advised by the upper bound which is equal to $n/k$. This approach allows for a wider range of values with the cost of a greater variance. 

Nonetheless, while the dimension is arbitrary, the empirical estimator of the extremal coefficient in \eqref{eq:emp_seco_sole} is consistent and the asymptotic deviation is well-understood (see, for instance, \cite{drees1998best,10.1214/12-AOS1023}) in the standard assumption of independent observations. In \ref{proof:prop_consistency}, Proposition \ref{prop:consistency}, we present arguments regarding the consistence of the proposed estimator of the $\SECO$ which goes beyond this classical setup of independent observations.

If $\textbf{Z}^{(a)}$ and $\textbf{Z}^{(b)}$ are asymptotically independent, there is no guarantee that an extreme event in one vector will be accompanied by an extreme event in the other vector, then the statistic in \eqref{eq:emp_seco} will converge in probability to zero (see \ref{proof:prop_consistency}, Proposition \ref{prop:consistency}). On the other hand, if $\textbf{Z}^{(a)}$ and $\textbf{Z}^{(b)}$ are asymptotically comonotone, then the SECO reduces to $\hat{\theta}(a) = \hat{\theta}(b)$ almost surely, since an extreme event in one vector will always be accompanied by an extreme event in the other vector. Additionally, the lower and upper bounds of $\widehat{\SECO}(a,b)$ is given by:
\begin{equation}
\label{eq:upper_bound_emp_seco}
 0 \leq \widehat{\SECO}(a,b) \leq \min \{ \hat{\theta}(a), \hat{\theta}(b) \} \quad \textrm{a.s.},
\end{equation}
where the upper one is reached when $\textbf{Z}^{(a)}$ and $\textbf{Z}^{(b)}$ are asymptotically comonotone. The resulting matrix, denoted by $\hat{\Theta}$, is a $d \times d$ matrix where each entry is given by
\begin{equation}
\label{eq:emp_seco_matrix}
\hat{\Theta}(a,b) = \widehat{\SECO}(a,b) / \min\{ \hat{\theta}(a), \hat{\theta}(b) \}, \quad a,b \in \{1,\dots,d\},
\end{equation}
which is the normalised $\SECO$ metric.

\begin{algorithm}
    \renewcommand{\thealgorithm}{(CAICE)}
	
	\caption{Clustering procedure for AI block models with compound extreme}

\begin{algorithmic}[1]
\Procedure{CAICE}{$S$, $\tau$, $\hat{\Theta}$}
    \State Initialise: $S = \{1,\dots,d\}$, $\hat{\Theta}(a,b)$ for $a,b \in \{1,\dots,d\}$ and $l = 0$
    \While{$S \neq \emptyset$}
    		\State $l = l +1$
    		\If{ $|S| = 1$}
    			\State $\hat{O}_l = S$
    		\EndIf
    		\If{$|S| > 1 $}
    			\State $(a_l, b_l) = \arg \underset{a,b \in S}{\max} \, \hat{\Theta}(a,b)$
    			\If{$\hat{\Theta}(a_l,b_l) \leq \tau$}
    				\State $\hat{O}_l = \{a_l\}$
    			\EndIf
    			\If{$\hat{\Theta}(a_l,b_l) > \tau$}
    				\State $\hat{O}_l = \{ s \in S : \hat{\Theta}(a_l,s) \wedge \hat{\Theta}(b_l,s) \geq \tau \}$
    			\EndIf
    		\EndIf
    		\State $S = S \setminus \hat{O}_l$
    	\EndWhile
    \State return $\hat{O} = (\hat{O}_l)_l$
\EndProcedure
\end{algorithmic}
\label{alg:rec_pratic}
\end{algorithm}

The algorithm that we present in this section takes as input the matrix $\hat{\Theta}$ in \eqref{eq:emp_seco_matrix}. This enables the division of $d$ objects of interest into the thinnest partition possible such that mutual asymptotic independence holds between clusters. Algorithm \ref{alg:rec_pratic} summarises the procedure for clustering asymptotically independent compound extreme events. As detailed in \cite{boulin2023high} for the Algorithm ECO which is similar to the above algorithm \eqref{alg:rec_pratic}, the overall complexity of the estimation procedure is $O(d^2 (d \vee  n \ln (n))))$.

The $\tau$ threshold is the only hyper-parameter in the \ref{alg:rec_pratic} Algorithm, and its selection is important in obtaining an effective partitioning. A useful tool for choosing an appropriate threshold is the $\SECO$ value for the resulting partition, which has been recommended in \cite{boulin2023high}, Section 3.4. This metric measures the divergence between the sum of the extremal coefficients of each cluster and the extremal coefficient of the entire vector (see \eqref{eq:seco} with group of different sizes). An effective partitioning is achieved when this metric is minimised, ideally for a moderate value of the threshold $\tau$. By identifying the threshold that results in the lowest $\SECO$ value, one can establish a partition of pixels, ensuring that their clusters are asymptotically independent from each other.

\section{Detecting concomitant extremes of compound precipitation and wind speed extremes}
\label{sec:results}

\subsection{Non-serially independent}
\label{subsec:non-serially}

To statistically assert departures from serial independence of multivariate time series, we conducted a randomness test as proposed by \cite{Genest2004} and \cite{ghoudi2001nonparametric}. We use the methodology presented in \cite{kojadinovic2009tests} and extended in \cite{kojadinovic2011tests} to detect serial dependence in continuous multivariate time series. Briefly, let $\textbf{Z}_1, \textbf{Z}_2, \dots$ be $d$-dimensional random vectors. We chose an integer $k > 1$, and for $\textbf{u} \in [0,1]^{dk}$, the vector $\textbf{u}_{\{j\}}$ is defined as follows:

\begin{equation*}
u^{(i)}_{\{j\}} = (u^{(i)} - 1) \mathds{1}_{ \{ i \in \{(j-1)q+1,\dots,jq \} \} } + 1,
\end{equation*}

We form the $dk$-dimensional random vector $\mathfrak{Z}_i = (\textbf{Z}_i,\dots, \textbf{Z}_{i+k-1})$, with $i = 1,\dots,n$. The serial independence empirical copula process in the multivariate setting thus write
\begin{equation}
    \label{eq:serial_indep_emp_cop_pro}
    \sqrt{n} \left( C_n^{s}(\textbf{u}) - \Pi_{j=1}^d C_n^{s}(\textbf{u}_{\{j\}}) \right),
\end{equation}
where $C_n^s$ is the serial empirical copula process computed with $\mathfrak{Z}_1,\dots,\mathfrak{Z}_n$. Under the hypothesis of serial independence, one can establish the asymptotic behavior of \eqref{eq:serial_indep_emp_cop_pro} to a tight Gaussian process. To obtain potentially powerful test obtained above from the empirical process, \cite{kojadinovic2011tests} derived $2^{k-1}-1$ tests statistic based on a Möbius decomposition of the process in \eqref{eq:serial_indep_emp_cop_pro} in the continuous multivariate time series setting. Those tests are implemented in the \textsf{copula} R package (\cite{yan2007enjoy}).

Due to computational limitations, we only considered three $3 \times 3$ pixels, at different locations, covering the initial $304$ days of the study, which represented the two first years of observation. We analysed precipitation and wind separately and the resulting dependogram is depicted in Fig. \ref{fig:dependogram} in \ref{sec:sup_fig} of the supplementary materials. It is notable that most of the ``subsets of lags" exhibited serial dependence. The function \textsf{multSerialIndepTest} computed three $p$-values, all of which provided robust evidence against serial independence. Furthermore, the decreasing trend of the computed statistic implied that the dependence was weakening for observations further apart. These results align with the mixing conditions stated in \ref{proof:prop_consistency}.
\begin{figure}[!h]

\centering
\label{subfig:map}\includegraphics[scale=0.4]{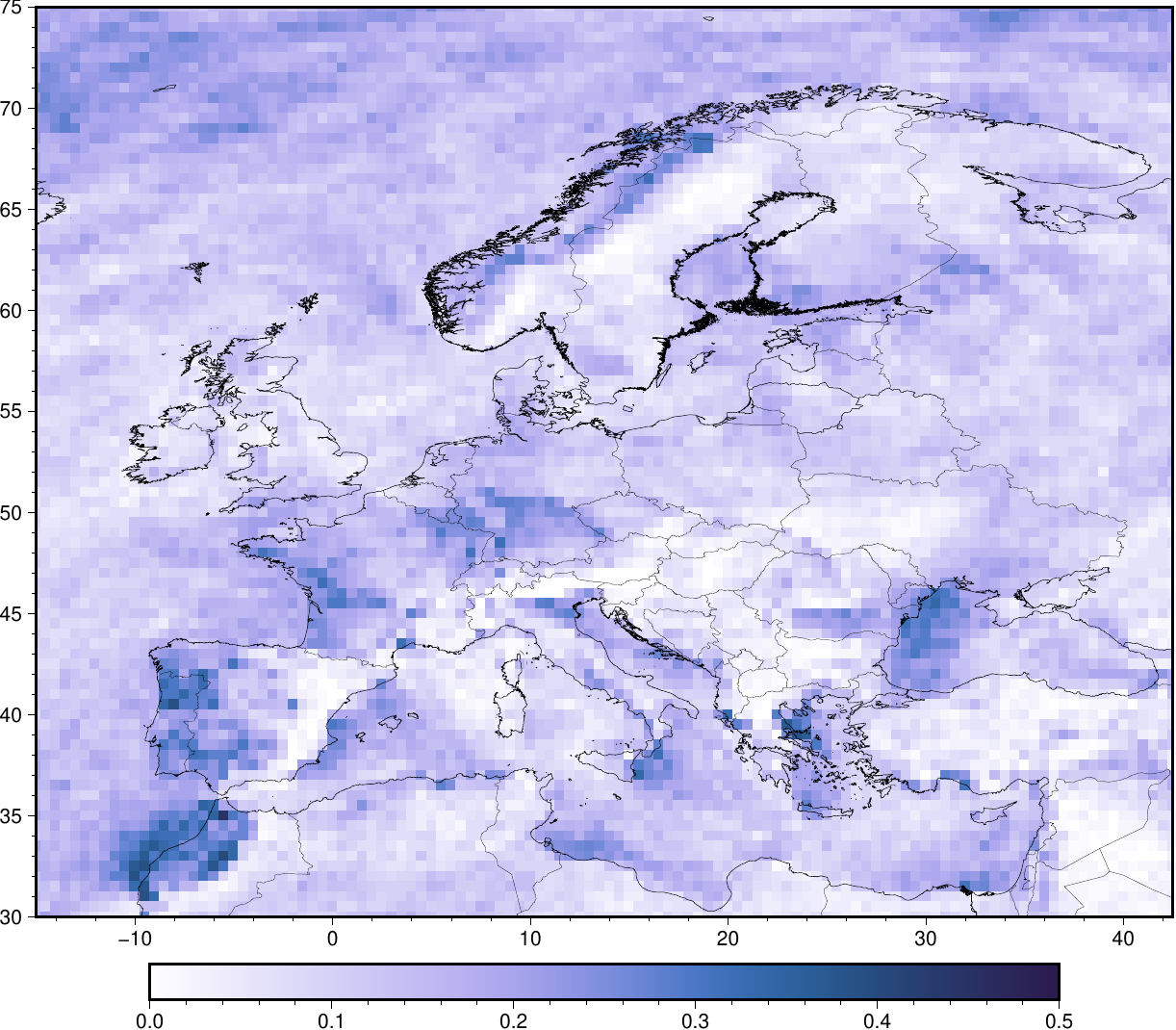}
\caption{Estimator of the extremal correlation within pixels, focusing on the daily sum of precipitation and wind speed maxima, $\hat{\chi}$ in Equation \eqref{eq:emp_ext_co}, maps for the $100$ largest values ($k=100$).}
\label{fig:main}
\end{figure}

\subsection{Exploratory analysis}

Extreme values of daily precipitation and wind speed maxima may occur together, and we aim to understand the spatial variability of this relationship. To identify regions for which extrema are non-concomitant between them, we consider the peak over threshold approach, which are values that lie above a certain value (see, for instance, \cite{beirlant2006statistics, haan2006extreme,resnick2007heavy} for an overall introduction to classical statistical methods in EVT). We conduct an exploratory analysis of the extremal dependence structure between the two variables within pixels. The estimated $\chi$ coefficient is defined by
\begin{equation}
    \label{eq:emp_ext_co}
    \hat{\chi}(a) = \frac{1}{k} \sum_{i=1}^n \mathds{1}_{ \{ R_i^{(a,1)} > n-k + 0.5, R_i^{(a,2)} > n - k +0.5 \} },
\end{equation}
where $R_i^{(a,\ell)}$ denotes the rank  of $Z_i^{(a,\ell)}$ among $Z_1^{(a,\ell)}, \dots, Z_n^{(a,\ell)}$, $\ell =1,2$. The estimated extremal coefficient between the variables reveals the highest co-occurrences in regions along the western coasts of Portugal, Spain, France, the UK, and Norway, as well as the northeastern coast of the Mediterranean (Fig. \ref{fig:main}). Conversely, the smallest co-occurrences are observed on the eastern coasts of the UK, Sweden, and Spain, over the northwestern coast of the Mediterranean, and around the Carpathian and southeastern Norwegian mountain ranges. These results are consistent with prior research (see, e.g., \cite{martius2016global}).

The low co-occurrence over eastern Norway and Spain may be attributed to the orographic enhancement of rain on the windward side of a mountain and the drying of the air as it reaches the lee (\cite{martius2016global}). This could also explain the high co-occurrences over the eastern coasts of the Mediterranean, where cyclones from the Mediterranean storm track may arrive perpendicularly to the mountains on the western coast of Italy and the eastern coasts of the Adriatic Sea (\cite{owen2021compound}). Additionally, the Cierzo winds may be responsible for the low co-occurrence to the south of the Pyrenees (\cite{martius2016global}). 

Nevertheless, the aforementioned analysis only considers extreme behavior within individual pixels. To dig deeper, we take a further step by using the empirical $\SECO$ outlined in Equation \eqref{eq:emp_seco} to explore interactions between pixels. Noteworthy, the empirical $\SECO$ does not inherently include spatial distances between pixels. As shown in Fig. \ref{fig:seco_dist}, we observe strong or moderate dependence between extremes for locations that are close to each other, while locations that are far apart have a normalised $\SECO$ near zero, indicating that extremes are weakly dependent or independent. The number of points per pixel highlights that most pixels are widely separated, with Euclidean distances between 20 to 40, and have small $\SECO$ values, indicating weak dependence or independence between compound extremes of two pixels. Therefore, even though it is not explicitly designed for this purpose, in the case of this specific dataset, the empirical $\SECO$ captures significant spatial information. This information could be highly valuable for comprehending the spatial patterns of precipitation and wind speed in Europe.

\begin{figure}[!h]
    \centering
    \includegraphics[scale=0.6]{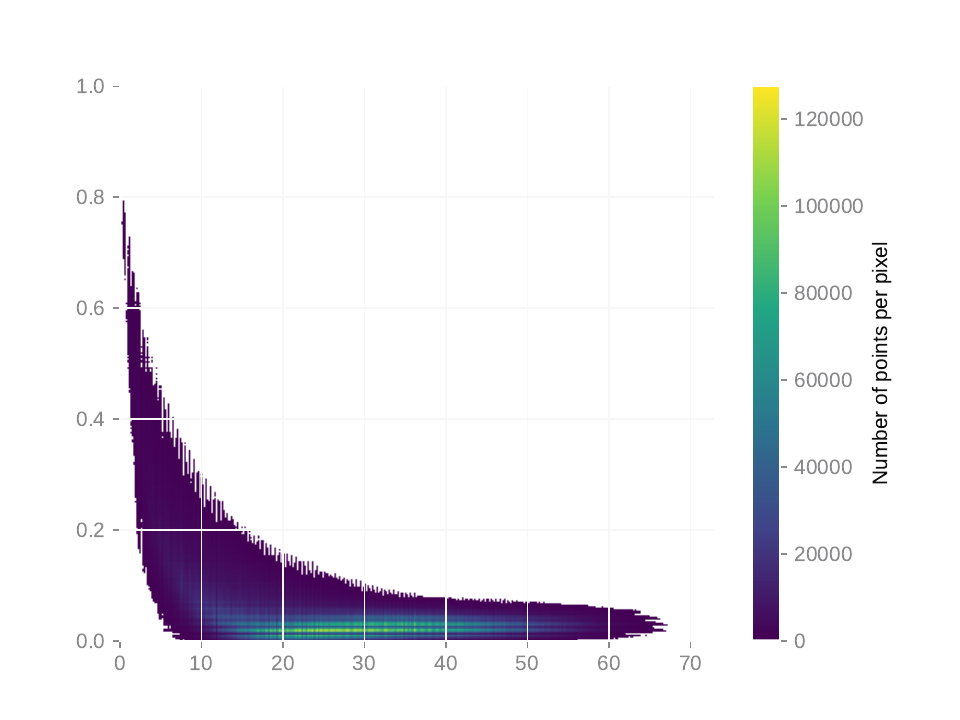}
    \caption{Pairwise $\SECO$ as a function of distance between sites.}
    \label{fig:seco_dist}
\end{figure}

In this context, our aim is to gain a deeper understanding of the spatial patterns of extreme total precipitation and wind speed maxima across Europe. To broaden our analysis of extremal pattern within pixels, we propose employing the clustering algorithm described in \ref{alg:rec_pratic} to group pixels while considering the extremal dependence of precipitation and wind speed between pixels.

\subsection{Clustering with constrained AI block models}

To delineate Europe into distinct regions that are mutually independent in compound weather extremes, we employ the clustering algorithm \ref{alg:rec_pratic}. Additionally, we employ a data-driven approach for selecting a suitable threshold, as elaborated below. Our initial step involves working with the matrix $\hat{\Theta}$ derived from Equation \eqref{eq:emp_seco_matrix}. The resultant clustering outcome is contingent on the value of $\tau$. Once Algorithm \ref{alg:rec_pratic} has been applied, we will denote the resulting partition as $\hat{O}(\tau) = \{\hat{O}_g\}_{g=1}^G$ of the set $\{1, \ldots, d\}$. Each cluster is characterised by a respective cardinality of $d_g$. As depicted in Figure \ref{fig:seco_dist}, even a small $\tau$ value within Algorithm \ref{alg:rec_pratic} can effectively partition the sub-regions of Europe. Nevertheless, a more rigorous approach to selecting the threshold value $\tau$ need to be discussed.

Let us introduce the random vector $\textbf{Y}^{(g)}_\tau$, composed of the variable index within $\hat{O}_g$. This means that $\textbf{Y}^{(g)}_\tau$ is a $(2d_g)$-dimensional random vector, considering two variables for each pixel in our case study. Moving forward, we can define the empirical $\SECO$ for groups of random vectors, which might have varying sizes within the given partition, using the following equation:
\begin{equation}
\label{eq:seco_partition}
\widehat{\SECO}(\textbf{Y}^{(1)}_\tau,\dots,\textbf{Y}^{(G)}_\tau) = \sum_{g=1}^G \hat{\theta}(g) - \hat{\theta}(1,\dots,d).
\end{equation}
The estimator mentioned above varies with $\tau$ due to its reliance on the partition $\hat{O}(\tau)$. As outlined in Section 3.4 of \cite{boulin2023high}, the values of $\tau$ that minimise $\SECO$ also ensure consistent recovery of our groups. For more details, we refer to Proposition 3 in \cite{boulin2023high}. To address this, we construct a loss function, denoted as $L$, over a grid of $\tau$ values denoted as $\Delta$, as follows:
\begin{equation}
\label{eq:loss_function}
L(\tau) = \ln \left(1 + \left(\widehat{\SECO} \left( \textbf{Y}^{(1)}_\tau,\dots,\textbf{Y}^{(G)}_\tau \right) - \underset{\tau \in \Delta}{\min} \, \widehat{\SECO}\left( \textbf{Y}^{(1)}_\tau,\dots,\textbf{Y}^{(G)}_\tau \right) \right) \right), \quad \tau \in \Delta.
\end{equation}

\begin{figure}[!h]

\begin{minipage}{.5\linewidth}
\centering
\subfloat[]{\label{subfig:seco_compound}\includegraphics[scale=.45]{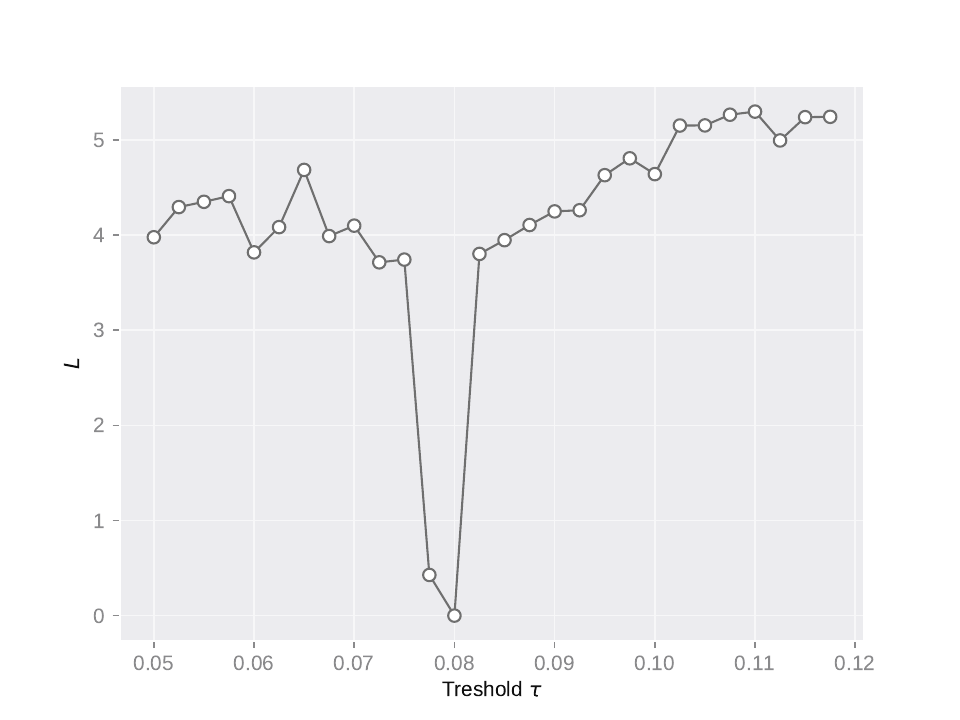}}
\end{minipage}%
\begin{minipage}{.5\linewidth}
\centering
\subfloat[]{\label{subfig:clust_compound}\includegraphics[scale=.45]{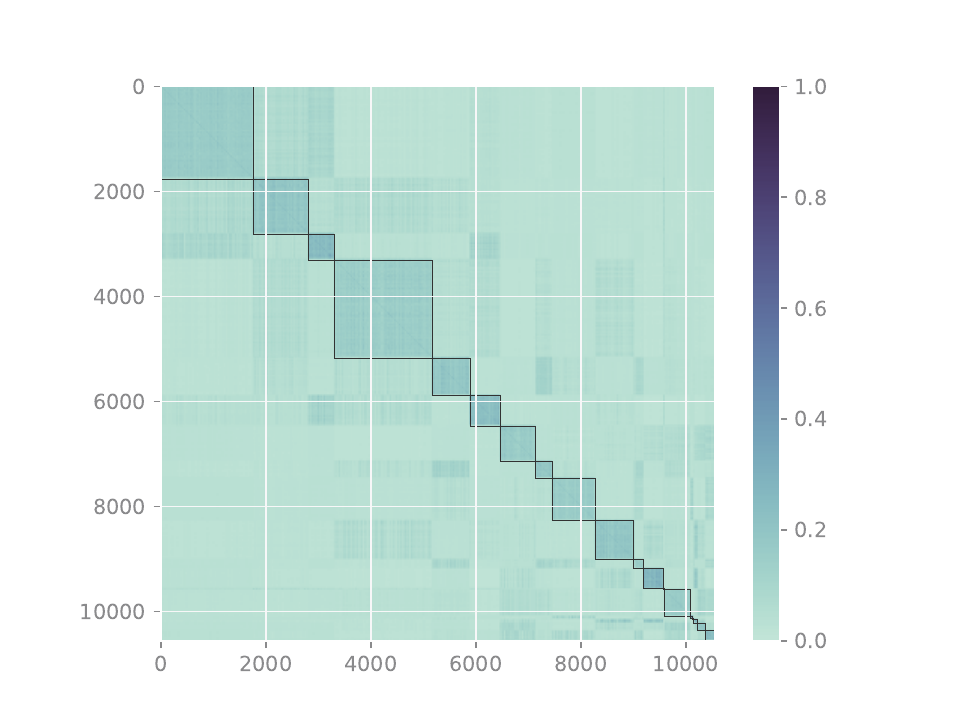}}
\end{minipage}

\caption{Value of the function $L(\tau)$ in \eqref{eq:loss_function} for different values of $\tau \in \Delta = \{0.05,0.0525,\dots,0.12\}$ in Panel \subref{subfig:seco_compound}. Partitions of the $\SECO$ similarity matrix with threshold $\tau = 0.08$ in Panel \subref{subfig:clust_compound}. Squares represent the clusters of variables.}
\label{fig:result_clust_compound}
\end{figure}

In our current high-dimensional context, where we have $n=6655$ daily observations and $d=10556$ pixels ($n < d$), estimating the extremal coefficient for the entire vector $\theta(1,\dots,d)$ encounters an upper bound that prevents it from reaching its theoretical maximum of $d$, as discussed in Section \ref{subsec:clt_comp_ext}. However, the loss function $L$ in equation \eqref{eq:loss_function} remains unaffected by this bias since it is mitigated by the subtraction operation. Nonetheless, it is essential to be mindful of this bias when computing $L$ for partitions with larger clusters, which typically occurs for smaller values of $\tau$. Therefore, we recommend reducing the number of extremes to $k=30$ when evaluating $L$. This adjustment widens the range of values used for estimating the extremal coefficient from higher values, but it comes at the expense of increased variance in the estimation process.

The value of $L$ with different values of $\tau$ with $\Delta = \{0.05,0.0525,\dots,0.12\}$ suggests that the best partitioning is found for $\tau = 0.08$ (Fig. \ref{fig:result_clust_compound}, Panel \subref{subfig:seco_compound}). For this threshold, we obtain $22$ clusters with $3$ of them having less than $10$ entities. We report the clustered $\SECO$ matrix defined in \eqref{eq:emp_seco_matrix} in Fig. \ref{fig:result_clust_compound}, Panel \subref{subfig:clust_compound}. 

\subsection{Results}

In Fig. \ref{fig:clust_1_9}, we present the twelve largest clusters obtained with the partition setting $\tau = 0.08$. Our algorithm effectively identifies sub-regions with strong dependence within clusters, as well as \textit{near-independence} or independence among compound extremes of daily precipitation and wind speed maxima in different clusters. The most prominent cluster, the fourth one, encompasses the North Sea and the Baltic Sea, which are connected basins. The North Sea is a marginal basin of the North Atlantic, and a shallow connection to the Baltic Sea exists through the Skagerrak and Kattegat regions. This area is particularly important for climate sciences and hydrology, and has inspired several works (see, for example, \cite{andree2022simulating, groger2019summer, wang2015development}). Another interesting area is the ninth cluster, consisting of the southwestern Black Sea, the Levantine Basin, and the southeast of the Anatolian peninsula. The Levantine Basin is known to be one of the windiest areas of the Mediterranean Sea, and the spatial distribution of the Levantine Basin in Figure 2 of \cite{soukissian2018offshore} outlines the geometry of this cluster. Additionally, extreme storm situations typically occur in December-January in the southwestern part of the Black Sea (\cite{divinsky2020extreme}). To the best of our knowledge, there are no studies on the tail dependence structure of climate variables in the southwestern Black Sea and the Levantine Basin. Furthermore, our algorithm sometimes distinguishes between the extremal behavior of land and sea, as illustrated by the first and sixth clusters for the Norwegian Sea and the Atlantic Ocean, respectively.

\begin{figure}[!h]
    \centering
    \includegraphics[width=11.5cm, height=8cm]{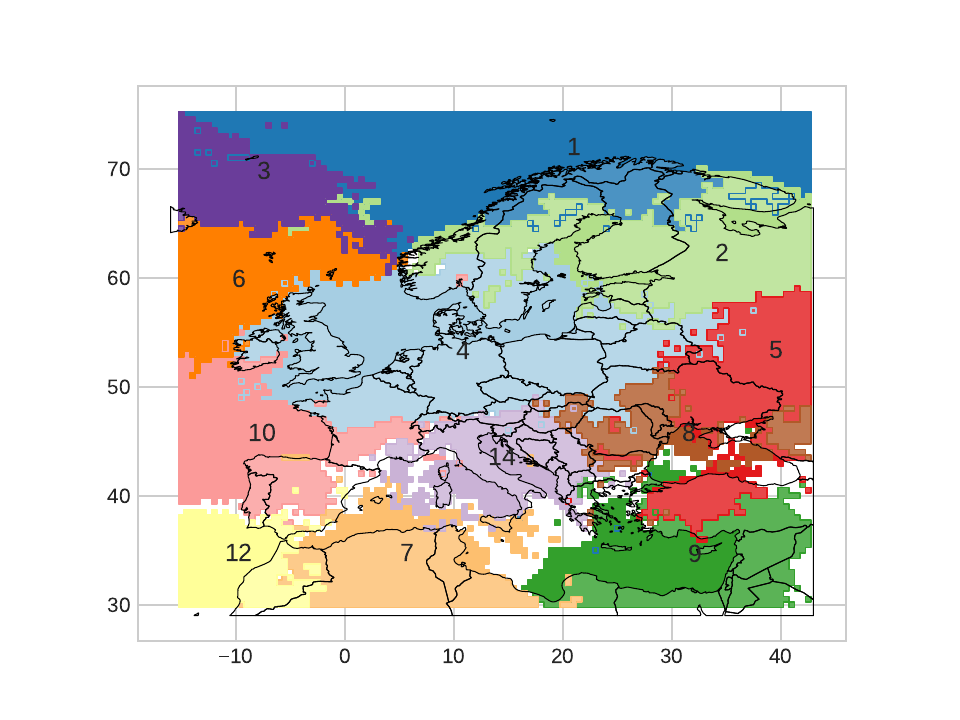}\vspace{-0.5cm}
    \caption{Representation of the $12$ largest clusters of the partition of the $\SECO$ matrix between pixels with Algorithm \ref{alg:rec_pratic} and threshold $\tau = 0.08$. Number of each cluster is depicted inside the cluster.} 
    \label{fig:clust_1_9}
\end{figure}

Identifying regions with simultaneous extreme events can be valuable for statistical modeling, and the method proposed here can address this question effectively. However, it is important to understand whether wind or precipitation is more important for the resulting clustering. To investigate this, a similar clustering process was applied separately to daily total precipitation and wind speed maxima using the extremal correlation as a dissimilarity. The resulting partitions are denoted by, $\hat{O}^{\textrm{P}}$, and $\hat{O}^{\textrm{W}}$ with an adapted threshold $\tau$. We will denote by $\hat{O}^{\textrm{PW}}$ the resulting partition for compound daily total precipitation and wind speed maxima. The figures depicting the results for the individual variables can be found in Fig. \ref{fig:result_clust_tp_wind} in \ref{sec:sup_fig}. 

The adapted threshold for the partition $\hat{O}^{P}$ is $\tau =0.09$, resulting in $70$ medium-sized clusters. For daily wind speed maxima, the optimal clustering $\hat{O}^{W}$ is obtained by setting $\tau = 0.07$, resulting in $24$ clusters. Upon visual inspection, our algorithmic partitioning reveals mosaic block patterns along the diagonal, while no clear patterns could be discerned from the off-diagonal. Additionally, the off-diagonal showcases moderate asymptotic dependence between groups or asymptotic independence, indicating that the resulting clustering aligns with the purpose of AI block models. The clusters for daily wind speed maxima are larger than those for daily total precipitation, which supports previous studies indicating that heavy gusts have a larger spatial impact than precipitation events (see, for example, \cite{pfahl2012spatial, raveh2015large}). With regard to spatial precipitation, several studies have shown that dependence tends to weaken for the largest observations (\cite{lalancette2021rank,le2018dependence}). This knowledge could explain why there are a large number of clusters with only a few entities for the clustering of daily total precipitation.

To compare different clustering methods, we use the Adjusted Rand Index (ARI), a popular measure used in clustering analysis (see, for instance, \cite{hubert1985comparing,rand1971objective}). To summarise, the ARI gives a concordance score between two different partitions. It takes values between 0 and 1 and the closer to 1, the more similar the partitions. For more details about its computation, we refer the reader to \ref{sec:ARI} of the supplementary materials. Computing the ARI between $\hat{O}^{PW}$ and $\hat{O}^{W}$ (resp, $\hat{O}^{PW}$ and $\hat{O}^{P}$), we obtain
\begin{equation*}
    \textrm{ARI}\left(\hat{O}^{PW}, \hat{O}^{W}\right) = 0.5, \quad \textrm{ARI}\left(\hat{O}^{PW}, \hat{O}^{P} \right) = 0.3.
\end{equation*}
An ARI value of $0$ indicates that the two sets are completely random and have nothing in common, while a value of $1$ indicates a perfect match between the two partitions. In this case, the ARI value of $0.5$ between the clustering of compound daily total precipitation and wind speed maxima and the clustering of the sole wind speed suggests that there is moderate similarity between the two sets, implying that there are fewer matching data points or clusters between them. In light of these results, we can conclude that the clustering of compound extreme is induced by both variables with a little more emphasis driven by wind speed maxima. 

\subsection{Alternative clustering method using $\SECO$}





In equation \eqref{eq:emp_seco_matrix}, $\hat{\Theta}$ can be envisioned as a similarity matrix, and $1-\hat{\Theta}$ as a dissimilarity matrix. This dissimilarity metric is smaller (or larger) for compound extremes that are dependent (or independent) between two pixels. The upper bound of the dissimilarity metric is $1$, which is reached when the compound extremes are independent. Thus, it is possible to perform classical method of clustering using this dissimilarity matrix. In particular, we have chosen to explore two different methods: a quantization-based approach and a Hierarchical clustering. These two approach requiring a specification of the number of clusters (unknown in practice), we use the ``silhouette coefficient" developed by, \cite{rousseeuw1987silhouettes} and which compares the tightness of clusters with their dissociation. In practice, the number of clusters, denoted as $K$, is determined by selecting the maximum average silhouette coefficient.

Quantization, also known as lossy data compression in information theory, involves the task of substituting data with an efficient and compact representation that allows for the reconstruction of the original observations with a certain degree of accuracy. A clustering problem can be viewed as an optimal quantization process aimed at minimising a specific loss function (for example, see \cite{banerjee2005clustering} or \cite[p. 15]{linder2002learning}). To generate clusters through the optimal quantization process, we utilised the algorithm described in references \cite{laloe20101} and \cite{laloe2021quantization}, setting arbitrary the number of clusters to $K = 10$ as the silhouette coefficient did not provide a suitable number of clusters. 

The process of hierarchical clustering begins with a basic partition, initially comprising of $d$ individual data points where each pixel stands alone as its own cluster. Subsequently, the data is grouped together incrementally, with clusters gradually merging until a single cluster encompassing all variables is achieved. At each step, the algorithm combines the two closest cluster centers according to a specific definition (for more details, refer to \cite[Chapter 12]{giraud2021introduction}), while keeping other clusters unchanged. For this approach, the silhouette coefficient leads us to set $K = 25$.

\begin{figure}[!h]

\begin{minipage}{.5\linewidth}
\centering
\subfloat[]{\label{subfig:clust_quanti}\includegraphics[scale=.45]{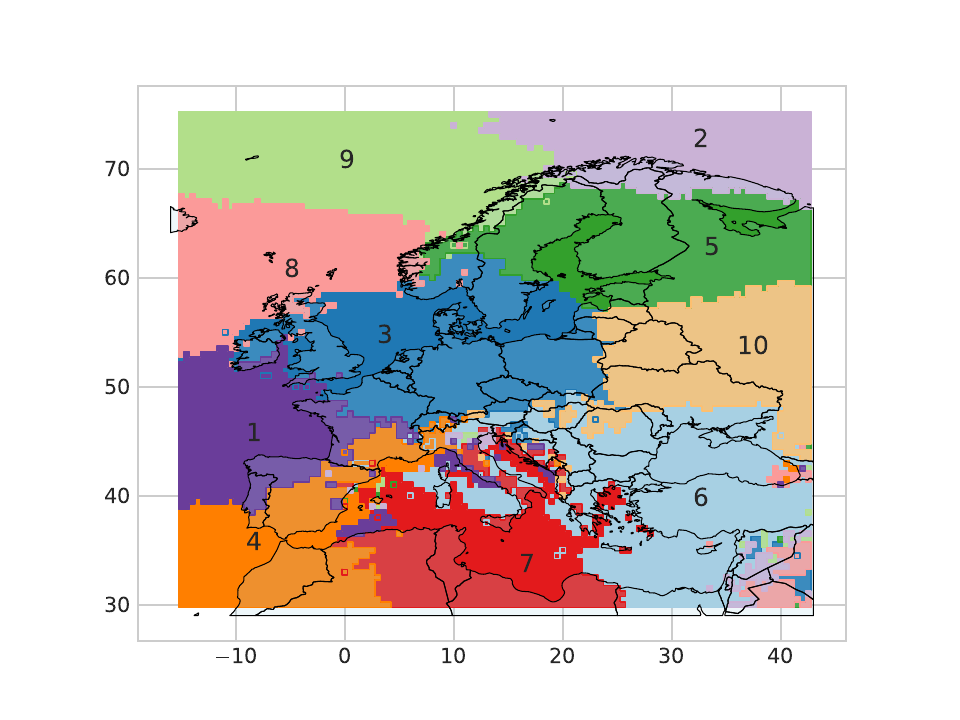}}
\end{minipage}%
\begin{minipage}{.5\linewidth}
\centering
\subfloat[]{\label{subfig:clust_hc}\includegraphics[scale=.45]{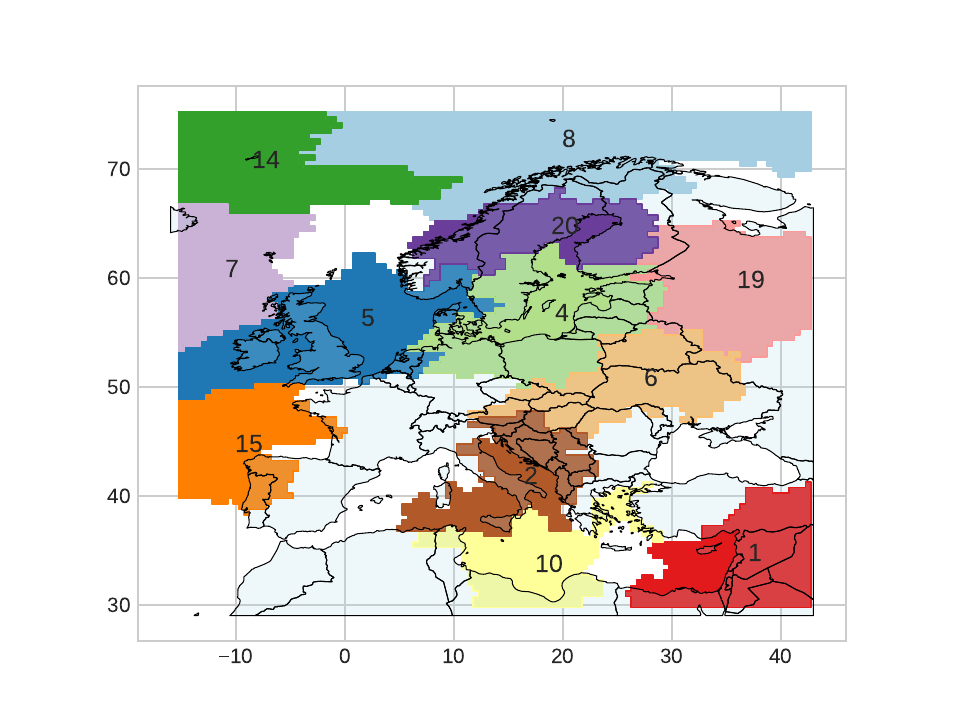}}
\end{minipage}

\caption{In Panel \subref{subfig:clust_quanti} is depicted the representation of the partition obtained through quantization-based approach with $K = 10$. In Panel \subref{subfig:clust_hc} a similar representation for the $12$ largest clusters is proposed through hierarchical clustering approach with $K=25$. These clusters were generated based on data related to daily total precipitation and wind speed maxima extremes from $10556$ pixels. Number of each cluster is depicted inside the cluster.}
\label{fig:clust_hc_quanti}
\end{figure}

The resulting partitions for this two methods are depicted in Figure \ref{fig:clust_hc_quanti}, and the clustered matrices can be found in Fig. \ref{fig:clust_quanti_matrix} and Fig. \ref{fig:result_clust_compound_hc} for further visualizations, in \ref{sec:sup_fig}. Notably, the clusters obtained through the quantization-based approach in the northern regions bear a resemblance to those obtained through Algorithm \ref{alg:rec_pratic}. For instance, clusters 3, 5, and 8 in Figure \ref{fig:clust_hc_quanti}, Panel \subref{subfig:clust_quanti}, are closely similar to clusters 2, 4, and 6 in Figure \ref{fig:clust_1_9} as identified through the \ref{alg:rec_pratic} Algorithm. Furthermore, as observed in the \ref{alg:rec_pratic} Algorithm, the quantization-based approach effectively distinguishes extreme behaviors between land and sea. This distinction is apparent in clusters 2 and 8 in Figure \ref{fig:clust_hc_quanti}, Panel \subref{subfig:clust_quanti}. These findings suggest the uniqueness of the $\SECO$ similarity measure in extracting valuable spatial information. One notable difference between the hierarchical clustering and Algorithm \ref{alg:rec_pratic} is that the clusters obtained through hierarchical clustering are smaller and depict stronger cross-dependencies. This phenomenon can be explained by the fact that some clusters which are separate in the hierarchical clustering partition are combined in the output of Algorithm \ref{alg:rec_pratic}. For instance, in Fig. \ref{fig:clust_hc_quanti}, Panel \subref{subfig:clust_hc}, most of pixels of clusters $4$ and $5$ in the hierarchical clustering are combined into a single cluster in Algorithm \ref{alg:rec_pratic}, the fourth one in Fig. \ref{fig:clust_1_9}. This suggests that compound extremes of those clusters are dependent, but smaller groups in the northern and Baltic Sea areas observe more concomitant extremes. Below, we investigate a hierarchical clustering of the fourth cluster given by Algorithm \ref{alg:rec_pratic} to inquiry whether or not we obtain a similar partition given by cluster $4$ and $5$ obtained in the hierarchical clustering.

\begin{figure}[!htb]
    \centering
    \includegraphics[width=11.5cm, height=8cm]{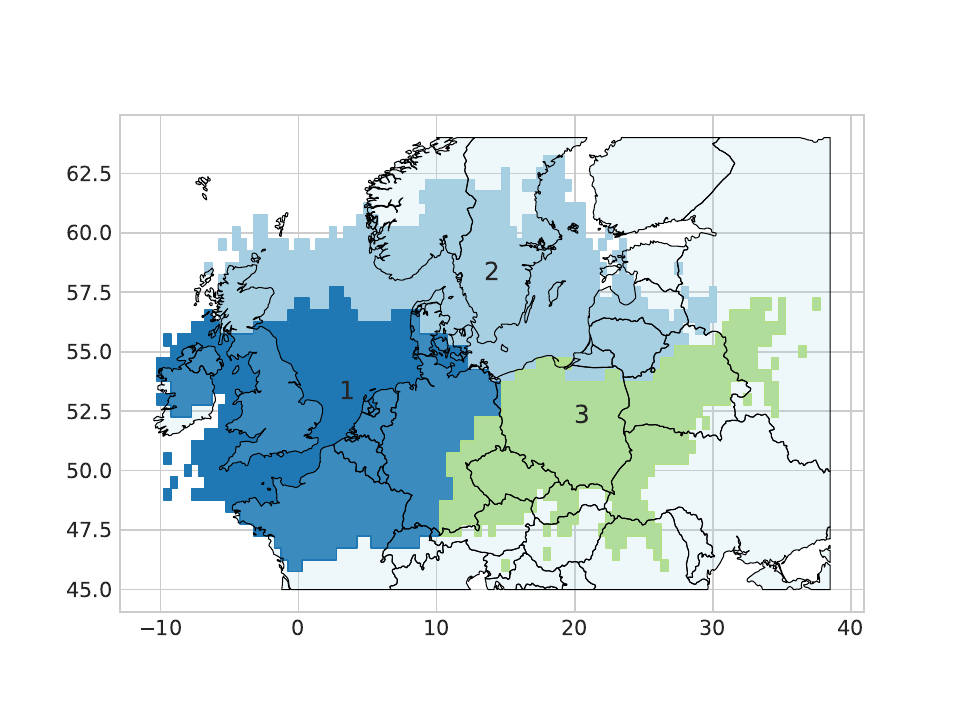}\vspace{-0.9cm}
    \caption{Representation of the $3$ clusters of the partition of the $1868$ pixels of the fourth cluster of the partition given by Algorithm \ref{alg:rec_pratic} using extremes of daily total precipitation and wind speed maxima. Number of each cluster is depicted inside the cluster.} 
    \label{fig:clust_thin}
\end{figure}

Fig. \ref{fig:clust_thin} displays the results of hierarchical clustering analysis on the most prominent cluster, i.e., cluster $4$ from the output of Algorithm \ref{alg:rec_pratic}. The analysis is performed on a reduced dataset, and the number of clusters is calibrated using the silhouette coefficient \cite{rousseeuw1987silhouettes}, with the maximum value obtained for $K=3$ (see Fig. \ref{fig:result_clust_thin}). The resulting matrix, $\hat{\Theta}$, reveals a strong dependence among compound extremes, indicating the presence of a whole asymptotic dependent vector with three distinct blocks with more concomitant compound extremes.

These results reveal intriguing patterns of co-occurring wind and precipitation extremes. The third cluster in Fig. \ref{fig:clust_thin}, which represents central-eastern Europe, is spatially coherent and having no access to the sea. The first cluster is connected to the North Sea and consists of western European countries (except Scotland). The second cluster comprises Scandinavian countries and the Baltic Sea, with a shallow connection to the North Sea. Notably, the North Sea's section belonging to this area corresponds to the Fladen Ground, the deepest part of the North Sea in the Scottish sector.

We also observe a similar partition given by the hierarchical clustering among all locations, where the North Sea is separated from the Baltic Sea. However, there is a difference in the Fladen Ground's classification, which is in different clusters for the hierarchical clustering among all pixels belonging to the North Sea cluster, see cluster $5$ in Fig. \ref{subfig:clust_hc}, and to the Baltic Sea cluster for the hierarchical clustering applied to the reduced dataset, see cluster 2 in Fig. \ref{fig:clust_thin}.

\section{Conclusion and perspectives}
\label{sec:conclusion}

Clustering spatial pixels is highly significant as it allows for a better comprehension of the inherent spatial pattern of a relevant physical phenomenon, enhances the accuracy of statistical procedures in situations where data is limited, and helps to identify regions where joint preventive measures can be taken to mitigate the impact of weather-related risks.

Traditionally, climate science research has focused on analysing single drivers or univariate dangers, which simplifies the complex dynamics of climate and its consequences. However, in reality, climate hazards often interact, leading to compound extremes. Therefore, statistical analyses that consider multiple hazards simultaneously are needed. In this paper, our goal is to identify subregions in Europe that demonstrate asymptotic independence concerning compound precipitation and wind speed extremes. In simpler terms, we want to find areas where two distinct subregions cannot experience concurrent compound extremes. To achieve this, we introduce a multivariate extreme value measure known as the $\SECO$ metric. This metric helps us quantify the extent to which random vectors of varying sizes deviate from asymptotic independence. We not only introduce this metric but also demonstrate the reliability of a non-parametric estimator for it, even in scenarios that go beyond the typical setup of independent observations. Building on this, we propose an algorithm specifically tailored for \emph{constrained} AI block model. This model ensures that pixels represent collections of univariate time series. Our algorithm allows us to pinpoint the largest partition where compound extremes exhibit independence over Europe. Interestingly, we uncover specific geographical patterns without relying on positional information, and it is worth noting that we do not need to pre-determine the number of clusters.

The proposed methodology can be extended to the case where the dataset has more than two variables, denoted as $p>2$. Furthermore, the methodology can also be extended to pixels with varying lengths of recorded time series, all the while preserving the concept of asymptotic independence. However, caution should be taken when interpreting the results in terms of coherence, or it may not be appropriate at all.

A major issue highlighted in this paper is how to estimate the dependence structure in high-dimensional datasets, where the number of variables is greater than the number of observations. As explained in \ref{sec:incompleteness}, traditional estimators are not able to accurately recover the dependence structure of an extreme value random vector if its margins are not sufficiently dependent, quantified by the condition $L(\textbf{x}) \geq n/k$ for $\textbf{x} \in \mathbb{R}^q$ where $n$ is the number of observations and $k$ the number of considered extremes. However, extreme value data are often scarce, making the high-dimensional setting common, except in certain cases.

\section*{Acknowledgments}

This work has been supported by the project ANR McLaren (ANR-20-CE23-0011). This work has been partially supported by the French government, through the 3IA Côte d’Azur Investments in the Future project managed by the  National Research Agency (ANR) with the reference number ANR-19-P3IA-0002. This work was also supported by the project ANR GAMBAS (ANR-18-CE02-0025) and by the french national programme LEFE/INSU.

\section*{Data Availability Statement}

Our code, $\SECO$ matrix and extremal correlation matrices are available from the following GitHub repository: \href{https://github.com/Aleboul/compound_extreme}{https://github.com/Aleboul/compound\_extreme}.

%% file: appendix_main.tex
\renewcommand\thefigure{\arabic{figure}}
\renewcommand\thetable{\arabic{table}}
\section{Axioms for a valid dependence measure}
\label{sec:axioms}

In this section, we recall the axiomatic framework to quantify dependence between multiple groups of random variables of possibly different sizes \cite{de2023parametric}. We recall that we consider a random vector $\textbf{Z}$ with distribution $F$ that is in the max domain of attraction of and EVD $H$. Plausible axioms, for a valid dependence measure of a random vector $\textbf{X}$, denoted as $\mathcal{D}(\textbf{X})$, are as follows

\begin{enumerate}[label=(A\arabic*)]
    \item For every permutation $\pi$ of $\textbf{X}^{(1)}$, \dots, $\textbf{X}^{(d)}$: $\mathcal{D}(\textbf{X}) = \mathcal{D}(\pi(\textbf{X}))$; and for every permutation $\pi_j$ of $X^{(j,1)},\dots,X^{(j,p_j)}$, for $j \in \{1,\dots, d\}$, it holds: $$\mathcal{D}(\textbf{X}) = \mathcal{D}(\textbf{X}^{(1)},\dots, \pi^{(j)}(\textbf{X}^{(j)}), \dots, \textbf{X}^{(d)}).$$ \label{axiom:A1}
    \item $0\leq \mathcal{D}(\textbf{X}) \leq 1$. \label{axiom:A2}
    \item $\mathcal{D}(\textbf{X}) = 0$ if and only if $\textbf{X}^{(1)},\dots,\textbf{X}^{(d)}$ are mutually independent. \label{axiom:A3}
    \item $\mathcal{D}(\textbf{X}^{(1)},\dots,\textbf{X}^{(d)},\textbf{X}^{(d+1)})$ with equality if and only if $\textbf{X}^{(d+1)}$ is independent of $(\textbf{X}^{(1)},\dots,\textbf{X}^{(d)})$.\label{axiom:A4}
    \item $\mathcal{D}(\textbf{X})$ is well defined for any $q$-dimensional random vector $\textbf{X}$ and is a functional of solely the copula $C$ of $\textbf{X}$.\label{axiom:A5}
    \item Let $T^{(j,\ell)}$ for $j = 1\dots,d$ and $\ell=1,\dots,p_j$ be strictly increasing, continuous transformations. Then
    \begin{equation*}
        \mathcal{D}(\textbf{T}^{(1)}(\textbf{X}^{(1)}),\dots,\textbf{T}^{(d)}(\textbf{X}^{(d)})) = \mathcal{D}(\textbf{X}^{(1)},\dots, \textbf{X}^{(d)}),
    \end{equation*}
    where $\textbf{T}^{(j)} = (T^{(j,1)}(X^{(j,1)}),\dots,T^{(j,p_j)}(X^{(j,p_j)}))$ for $j = 1,\dots,d$. \label{axiom:A6}
    \item Let $T^{(j,\ell)}$ be a strictly decreasing, continuous transformation for a fixed $j \in \{1,\dots,d\}$ and a fixed $\ell \in \{1,\dots,p_j\}$. Then
    \begin{equation*}
        \mathcal{D}(\textbf{X}^{(1)},\dots, T^{(j,\ell)}(\textbf{X}^{(j)}),\dots,\textbf{X}^{(d)}) = \mathcal{D}(\textbf{X}^{(1)},\dots,\textbf{X}^{(d)}),
    \end{equation*}
    where $T^{(j,\ell)}(\textbf{X}^{(j)}) = (X^{(j,1)},\dots, T^{(j,\ell)}(X^{(j,\ell)}),\dots,X^{(j,p_i)})$. \label{axiom:A7}
    \item Let $(\textbf{X}_n)_{n \in \mathbb{N}}$ be a sequence of $q$-dimensional reduction random vectors having copulas $(C_n)_{n \in \mathbb{N}}$, then
    \begin{equation*}
        \underset{n \rightarrow \infty}{\lim} \mathcal{D}(\textbf{X}_n) = \mathcal{D}(\textbf{X})
    \end{equation*}
    if $C_n \rightarrow C$ uniformly, where $C$ denotes the copula of $\textbf{X}$. \label{axiom:A8}
\end{enumerate}

Having those necessary materials, we now detail below which axioms hold for the $\SECO$ metric to measure the dependence among extreme of random vectors.

\begin{proposition}
    \label{prop:dep_meas}
    Let $\SECO$ be the metric defined in \eqref{eq:seco}. Then, it satisfies the system of axioms \ref{axiom:A1}, \ref{axiom:A2} with the following bounds
    \begin{equation*}
        0 \leq \SECO(\textbf{Z}^{(1)},\dots,\textbf{Z}^{(d)}) \leq \underset{j = 1,\dots,d}{\min}\left\{ \sum_{k \neq j} \theta(k) \right\},
    \end{equation*} \ref{axiom:A3}, \ref{axiom:A4}, \ref{axiom:A5}, \ref{axiom:A6} and \ref{axiom:A8} stated in \cite{de2023parametric}.
\end{proposition}

\begin{proof}
The reason why Property \ref{axiom:A1} holds is due to the fact that the set union and addition of numbers are commutative. Proposition 7 and 8 from \cite{boulin2023high} imply the results stated about \ref{axiom:A2} for the lower bound and \ref{axiom:A3}, where the latter is related to asymptotic independence. To obtain the upper bound, it is observed that
    \begin{equation*}
        0 \leq \theta(1,\dots,d) \leq \min\{\theta(1),\dots,\theta(d)\}.
    \end{equation*}
    For property (A4), the inequality
    \begin{equation*}
        \SECO(\textbf{Z}^{(1)},\dots,\textbf{Z}^{(d)}, \textbf{Z}^{(d+1)}) \geq \SECO(\textbf{Z}^{(1)},\dots,\textbf{Z}^{(d)})
    \end{equation*}
    can be rewritten as
    \begin{equation*}
        \theta(1,\dots,d) + \theta(d+1) \geq \theta(1,\dots,d,d+1),
    \end{equation*}
    which holds true by Proposition 7 of \cite{boulin2023high}. Moreover, this inequality holds as an equality if and only if $\textbf{Z}^{(d+1)}$ is asymptotically independent of $(\textbf{Z}^{(1)},\dots,\textbf{Z}^{(d)})$, according to Proposition 8 of \cite{boulin2023high}.

    We have for every $u \in (0,1)$,
    \begin{equation*}
        \theta(1,\dots,d) = \ln\left( C(u,\dots,u) \right) / \ln(u),
    \end{equation*}
    where $C$ is the extreme value copula in the domain of attraction of $\textbf{Z}$, i.e.,
    \begin{equation*}
        C(u^{(1,1)},\dots,u^{(d,p_d)}) = \exp\left\{-L\left(-\ln u^{(1,1)},\dots, -\ln u^{(d,p_d)}\right) \right\}.
    \end{equation*}
    Hence \ref{axiom:A5} and \ref{axiom:A6} are fulfilled. Let $(\textbf{Z}_n)_{n \in \mathbf{N}}$ be a sequence of random vector, for each $n$, suppose that $\textbf{Z}_n$ is in the max-domain of attraction of a random vector $H_n$ an extreme value distribution with extreme value copula $C_n$. Suppose that the sequence of extreme value copulae $(C_n)_{n \in \mathbb{N}}$ converges uniformly to $C$, where $C$ is the extreme value copula of $H$. Then, for every $\textbf{u} \in (0,1)$, we have the point wise convergence:
    \begin{equation*}
        C_n(\textbf{u}) \rightarrow C(\textbf{u}), \quad n \rightarrow \infty.
    \end{equation*}
    Thus, as $\ln : (0,1) \rightarrow \mathbb{R}_{-}$ is continuous, we have
    \begin{equation*}
        \ln (C_n(\textbf{u})) \rightarrow \ln (C(\textbf{u})), \quad n \rightarrow \infty.
    \end{equation*}
    Then
    \begin{equation*}
        \SECO(\textbf{Z}_{n}^{(1)},\dots, \textbf{Z}_n^{(d)}) \rightarrow \SECO(\textbf{Z}^{(1)},\dots,\textbf{Z}^{(d)}), \quad n \rightarrow \infty.
    \end{equation*}
    
\end{proof}

\section{A coherent measure for extreme value random vectors}
\label{sec:coherent_measure}

In this section, we borrow the notation previously introduced in the specific case where $d = 2$ and $p_a = p_b = p$. Let $\textbf{Z}$ and $\textbf{W}$ be two random vectors with marginal random vectors $\textbf{Z}^{(a)}$ and $\textbf{Z}^{(b)}$ (resp. $\textbf{W}^{(a)}$ and $\textbf{W}^{(b)}$), each having $p$ components. We set $q = 2p$. Assume that $\textbf{Z}$ and $\textbf{W}$ are in the max-domain of attraction of two extreme value distributions with stable tail dependence function $L_\textbf{Z}$ and $L_\textbf{W}$. Here, $L_\textbf{Z}$ and $L_\textbf{W}$ are nested, and for $\textbf{x} \in [0,\infty)^q$, $\textbf{x}^{(a)} = (x^{(a,1)},\dots,x^{(a,p)}), \textbf{x}^{(b)} = (x^{(b,1)},\dots,x^{(b,p)})$, we have
\begin{align*}
&L_\textbf{Z}(\textbf{x}) = L_\textbf{Z}^{(0)}\left(L_\textbf{Z}^{(a)}(\textbf{x}^{(a)}), L_\textbf{Z}^{(b)}(\textbf{x}^{(b)})\right), \quad L_\textbf{W}(\textbf{x}) = L_\textbf{W}^{(0)}\left(L_\textbf{W}^{(a)}(\textbf{x}^{(a)}), L_\textbf{W}^{(b)}(\textbf{x}^{(b)})\right),
\end{align*}
where $L_\textbf{Z}^{(0)}$ and $L_\textbf{Z}^{(j)}$ are the``mother" stable tail dependence function and the stable tail dependence functions associated with the extreme value random vectors $\textbf{X}^{(j)}$, respectively. Similarly, $L_\textbf{W}^{(0)}$ and $L_\textbf{W}^{(j)}$ represent the stable tail dependence functions associated with the extreme value random vectors $\textbf{Y}^{(j)}$ for $j = a,b$. The $\SECO$ for these two models thus reduces to
\begin{align*}
    &\SECO(\textbf{Z}^{(a)}, \textbf{Z}^{(b)}) = L_\textbf{Z}^{(a)}(\textbf{1}^{(a)}) + L_\textbf{Z}^{(b)}(\textbf{1}^{(b)}) - L_\textbf{Z}^{(0)}\left(L_\textbf{Z}^{(a)}(\textbf{1}^{(a)}), L_\textbf{Z}^{(b)}(\textbf{1}^{(b)})\right), \\
    &\SECO(\textbf{W}^{(a)}, \textbf{W}^{(b)}) = L_\textbf{W}^{(a)}(\textbf{1}^{(a)}) + L_\textbf{W}^{(b)}(\textbf{1}^{(b)}) - L_\textbf{W}^{(0)}\left(L_\textbf{W}^{(a)}(\textbf{1}^{(a)}), L_\textbf{W}^{(b)}(\textbf{1}^{(b)})\right)
\end{align*}
Taking advantage of the definition introduced by \cite{scarsini1984measures}, we say that $\textbf{Z}$ is more concordant than $\textbf{W}$ in the max-domain of attraction if for every possible value of $\textbf{x} \in \mathbb{R}^q$, $L^{(0)}_\textbf{Z}(\textbf{x}) \leq L^{(0)}_\textbf{W}(\textbf{x})$. Now, let's suppose that $\textbf{Z}$ and $\textbf{W}$ have the same marginal extremal dependence structure, which means that their stable tail dependence functions are the same for both components, i.e.,
\begin{equation*}
    L_\textbf{Z}^{(a)}(\textbf{x}^{(a)}) = L_\textbf{W}^{(a)}(\textbf{x}^{(a)}), \quad L_\textbf{Z}^{(a)}(\textbf{x}^{(a)}) = L_\textbf{W}^{(b)}(\textbf{x}^{(b)}), \quad \textbf{x} \in \mathbb{R}^q.
\end{equation*}
Thus, we have
\begin{align*}
    &\SECO(\textbf{Z}^{(a)}, \textbf{Z}^{(b)}) - \SECO(\textbf{W}^{(a)}, \textbf{W}^{(b)}) = \\ &L_\textbf{W}^{(0)}\left(L^{(a)}_\textbf{W}(\textbf{1}^{(a)}), L^{(b)}_\textbf{W}(\textbf{1}^{(b)})\right) - L_\textbf{Z}^{(0)}\left(L^{(a)}_\textbf{W}(\textbf{1}^{(a)}), L^{(b)}_\textbf{W}(\textbf{1}^{(b)})\right)
\end{align*}
which is positive. Thus the $\SECO$ is a coherent measure given the marginal random vector's dependence structure is the same. 

To better understand the behaviour of the SECO, we will consider two specific nested models: the nested Gumbel and Hüsler-Reiss models, for which $p = 2$ and $d = 2$. For these models, we will analyse the corresponding stable tail dependence functions.

\begin{align*}
    &L^{\textrm{Gu}}(x_1,x_2) = \left(x_1^{1/\alpha} + x_2^{1/\alpha}\right)^{\alpha}, \\ &L^{\textrm{HR}}(x_1,x_2) = \Phi\left( \frac{\lambda}{2} + \frac{x_2-x_1}{\lambda} \right)x_1 + \Phi \left( \frac{\lambda}{2} + \frac{x_1-t_2}{\lambda}\right) x_2,
\end{align*}
with $\alpha \in [0,1]$, $\lambda \in [0,\infty)$ and $\Phi$ denotes the cumulative distribution function of a standard normal random variable. The first stable tail dependence function is known as the Gumbel (or Logistic) distribution introduced by \cite{gumbel1960bivariate}, the parameter $\alpha$ represent the strength of dependence: if $\alpha \rightarrow 0$, then the two random variables are comonotone while if $\alpha = 1$ it reduces to independence. The second stable tail dependence function is the Hüsler  Reiss distribution (\cite{husler1989maxima}), both asymptotic comonotony and independence are depicted by the respective limits $\lambda \rightarrow 0$ and $\lambda \rightarrow \infty$. 

For the nested Gumbel model, the $\SECO$ is equal to

\begin{equation}
    \label{eq:nested_logistic_seco}
    2^{\alpha_a} + 2^{\alpha_b} - \left( 2^{\alpha_a / \alpha_0} + 2^{\alpha_b / \alpha_0} \right)^{\alpha_0}.
\end{equation}
This equation involves the parameters $\alpha_0, \alpha_a,$ and $\alpha_b$, which correspond to the mother, the first, and the second random vector margins, respectively. In Fig. \ref{fig:nested_gumbel_seco} are depicted level sets of \eqref{eq:nested_logistic_seco} and the normalised version, i.e., \eqref{eq:nested_logistic_seco} divided by $\min\{2^{\alpha_a}, 2^{\alpha_b}\}$, according to $\alpha_0$ for different dependence structure between the marginal random vectors. Additionally, we can also calculate the $\SECO$ for the Hüsler-Reiss nested model which equals to

\begin{equation}
    \label{eq:nested_hr_seco}
    \theta^{(a)} + \theta^{(b)} - \left[ \Phi\left( \frac{\lambda_0}{2} + \frac{\theta^{(b)}- \theta^{(a)}}{\lambda_0} \right)\theta^{(a)} + \Phi\left( \frac{\lambda_0}{2} + \frac{\theta^{(a)}- \theta^{(b)}}{\lambda_0} \right)\theta^{(b)} \right],
\end{equation}
where $\theta^{(j)} = 2 \Phi (\lambda_j / 2)$, $j = a,b$ are the extremal coefficients and includes several parameters such as $\lambda_a$, $\lambda_b$, and $\lambda_0$ which correspond to the first, the second random vector margins and the mother respectively. In Fig. \ref{fig:nested_hr_seco}, we represent level sets of \eqref{eq:nested_hr_seco} and the normalised version of \eqref{eq:nested_hr_seco}, that is divided by $\min\{ \theta^{(a)}, \theta^{(b)} \}$, the bound found in Proposition~\ref{prop:dep_meas}.
\begin{figure}[!ht]

\begin{minipage}{.5\linewidth}
\centering
\subfloat[]{\label{subfig:nested_gumbel}\includegraphics[scale=.5]{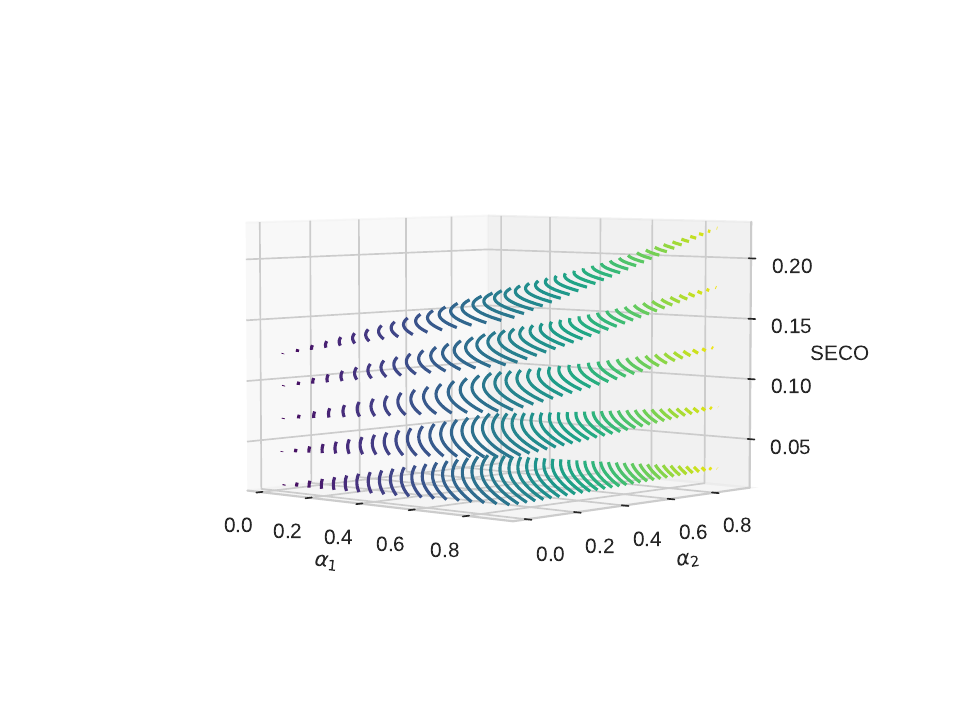}}
\end{minipage}%
\begin{minipage}{.5\linewidth}
\centering
\subfloat[]{\label{subfig:nested_gumbel_normalized}\includegraphics[scale=.5]{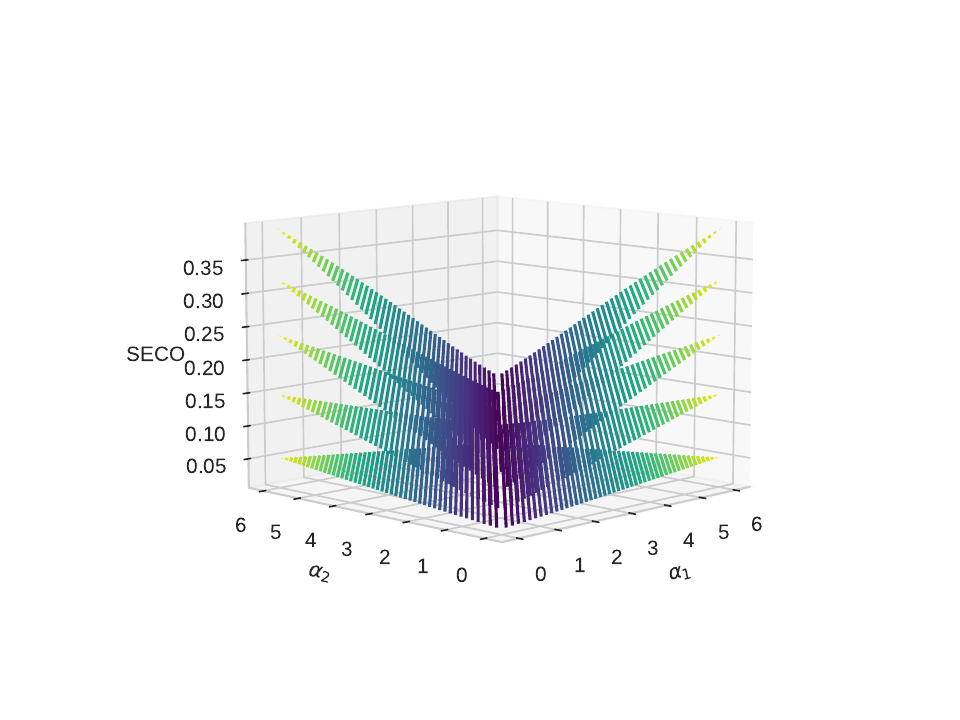}}
\end{minipage}

\caption{Level sets of the $\SECO$ for the nested Logistic model (see \eqref{eq:nested_logistic_seco}) in Panel \ref{subfig:nested_gumbel} and the normalised $\SECO$ in Panel \ref{subfig:nested_gumbel_normalized} for $\alpha_0 \in \{0.91,0.93,0.95,0.97,0.99\}$ and $\alpha_1, \alpha_2$ range in $\{0.01,0.02,\dots,0.9\}$. The coherence property ensures that the level sets are arranged in ascending order based on the values of $\alpha_0$.}
\label{fig:nested_gumbel_seco}
\end{figure}

\begin{figure}[!ht]

\begin{minipage}{.5\linewidth}
\centering
\subfloat[]{\label{subfig:nested_hr}\includegraphics[scale=.5]{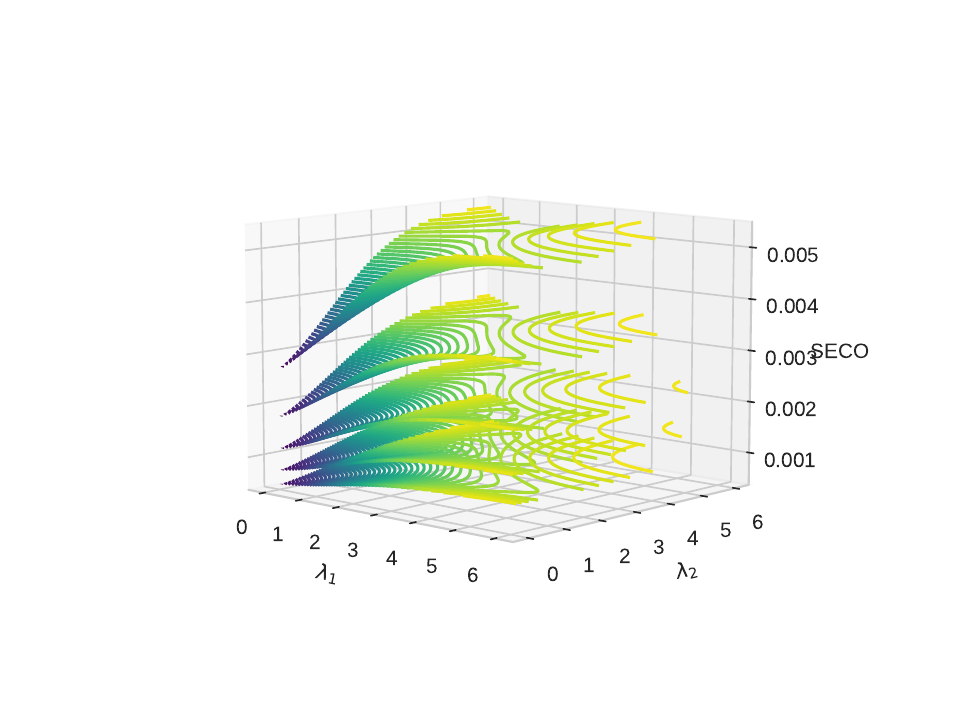}}
\end{minipage}%
\begin{minipage}{.5\linewidth}
\centering
\subfloat[]{\label{subfig:nested_hr_normalized}\includegraphics[scale=.5]{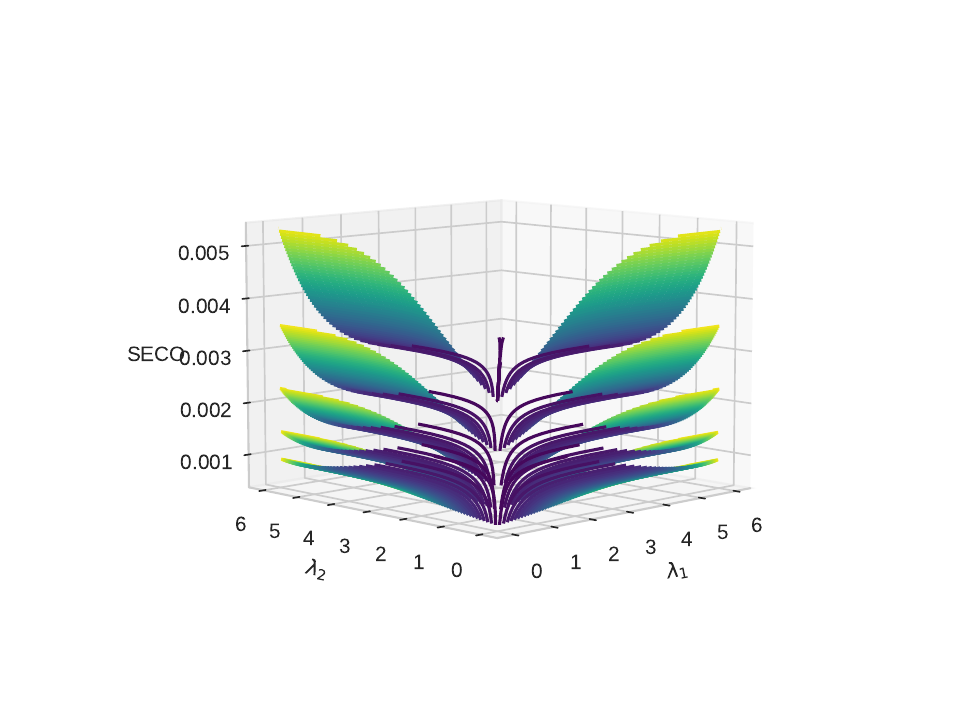}}
\end{minipage}

\caption{Level sets of the $\SECO$ for the nested Hüsler-Reiss model (see \eqref{eq:nested_hr_seco}) in Panel \ref{subfig:nested_hr} and the normalised $\SECO$ in Panel \ref{subfig:nested_hr_normalized} for $\lambda_0 \in \{6.0,6.25,6.5,6.75,7.0\}$ and $\lambda_1, \lambda_2$ range in $\{0.01,0.02,\dots,6.0\}$. The coherence property ensures that the level sets are arranged in ascending order based on the values of $\lambda_0$.}
\label{fig:nested_hr_seco}
\end{figure}

In both Fig. \ref{fig:nested_gumbel_seco} and Fig. \ref{fig:nested_hr_seco}, the behaviour of $\SECO$ exhibits similarities, particularly in relation to the level sets. These level sets are characterised by their monotonicity, whereby the level set's height increases with the degree of dependence among the marginal random vectors. In other words, the greater the dependence between the marginal random vectors, the higher the level set will be. Additionally, the $\SECO$ demonstrates monotonicity through $\alpha_1, \alpha_2$ (or alternatively, $\lambda_1, \lambda_2$), whereby the value of the $\SECO$ increases as the random variables in the marginal random vectors become more concordant. In simpler terms, the more closely related the random variables are, the higher the $\SECO$ will be. Moving on to the normalised $\SECO$, we observe that the highest values are attained when the random variables within the marginal random vectors display asymmetric behaviour, meaning that one variable is comonotonic while the other is independent. Conversely, the lowest values are obtained when both variables share the same dependence structure.

\section{Incompleteness tail dependence structure estimation in high dimension}
\label{sec:incompleteness}

Throughout the section, assume that we have $\textbf{Z}_1, \dots, \textbf{Z}_n$ identically distributed observations of the $d$-dimensional random vector $\textbf{Z}$, which is in the max-domain of attraction of $H$, an EVD where each of the components of $\textbf{Z}$ are asymptotically independent. Let $R_{n,i}^{(j)}$ denote the rank of $Z_i^{(j)}$ among $Z_1^{(j)},\dots,Z_n^{(j)}$, $i = 1,\dots,n$, $j=1,\dots,d$. For $k \in \{1,\dots,n\}$, define a nonparametric estimator of $\theta$, the extremal coefficient by

\begin{equation*}
    \hat{\theta}_{n,k}^{\textrm{EKS}} := \frac{1}{k} \sum_{i=1}^n \mathds{1}_{ \{ R_{n,i}^{(1)} > n + 0.5 - k \textrm{ or } \dots \textrm{ or } R_{n,i}^{(d)} > n + 0.5 - k\} },
\end{equation*}
see \cite{10.1214/12-AOS1023}.

\begin{lemma}
    \label{lem:inequality_EKS}
    For $k \in \{1,\dots,n\}$ we have
    \begin{equation*}
        \hat{\theta}_{n,k}^{\textrm{\normalfont{EKS}}} \leq \frac{n}{k}
    \end{equation*}
\end{lemma}
\begin{proof}
    The upper bound is trivial since
    \begin{equation*}
        \sum_{i=1}^n \mathds{1}_{ \{ R_{n,i}^{(1)} > n + 0.5 - k \textrm{ or } \dots \textrm{ or } R_{n,i}^{(d)} > n + 0.5 - k\} } \leq n.
    \end{equation*}
\end{proof}

Now, let us divide the sample of size $n$ of $\textbf{Z}$ into $k$ blocks of length $m$, so that $k = n / m$ (where we suppose, without loss of generality that m divide n). For the ith block, the maximum value in the $j$-component is denoted by
\begin{equation*}
    M_{m,i}^{(j)} = \max \{ Z_t^{(j)} : t \in (im -m, im] \}.
\end{equation*}
Let $R_{n,m,i}^{(j)}$ denote the rank of $M_{m,i}^{(j)}$ among  $M_{m,1}^{(j)}, \dots, M_{m,j}^{(j)}$, $i = 1,\dots,n$, $j = 1,\dots,d$. Define a non parametric estimator of the multivariate madogram 
\begin{equation*}
    \hat{\nu}_{n,m} := \frac{1}{k} \sum_{i=1}^k \left[ \bigvee_{j=1}^d \frac{R_{n,m,i}^{(j)}}{k+1} - \frac{1}{d} \sum_{j=1}^d \frac{R_{n,m,i}^{(j)}}{(k+1)} \right]
\end{equation*}
\begin{lemma}
    \label{lem:inequality_mado}
    For $m \in \{1,\dots,n\}$, we have
    \begin{equation*}
        \hat{\nu}_{n,m} \leq \frac{k}{k+1} - \frac{1}{2}.
    \end{equation*}
    Consequently,
    \begin{equation*}
        \hat{\theta}_{n,m}^{\textrm{\normalfont{MAD}}} \leq \frac{n}{m}.
    \end{equation*}
\end{lemma}
\begin{proof}
    One can easily deduce the following upper bound
    \begin{equation*}
        \bigvee_{j=1}^d \frac{R_{n,m,i}^{(j)}}{k+1} \leq \frac{k}{k+1}.
    \end{equation*}
    Thus we obtain that
    \begin{equation*}
        \hat{\nu}_{n,m} \leq \frac{k}{k+1} - \frac{1}{k} \frac{1}{d} \sum_{j=1}^d \sum_{i=1}^k \frac{R_{n,m,i}^{(j)}}{k+1}.
    \end{equation*}
    The right hand side of the equation is equal to
    \begin{equation*}
        \frac{1}{k} \frac{1}{d} \sum_{j=1}^d \sum_{i=1}^k \frac{R_{n,m,i}^{(j)}}{k+1} = \frac{1}{k} \frac{1}{d} \sum_{j=1}^d \frac{k(k+1)}{2(k+1)} = \frac{1}{2}.
    \end{equation*}
    Let us consider the following function
    \begin{align*}
        f \colon \left[0,\frac{k}{k+1}-\frac{1}{2}\right] & \rightarrow\ \mathbb{R} \\
        x&\mapsto \frac{0.5+x}{0.5-x}.
    \end{align*}
    Since it is an nondecreasing function, we must have
    \begin{equation*}
        \hat{\theta}_{n,m}^{\textrm{MAD}} \leq f\left(\frac{k}{k+1} - \frac{1}{2}\right) = \frac{n}{m}
    \end{equation*}
\end{proof}

The consequence of both lemmas is that the nonparametric and the madogram-based estimator of the extremal coefficient can only recover values in the following range:
\begin{equation*}
    1 \leq \hat{\theta}^{\textrm{EKS}}_{n,k} \leq \frac{n}{k}, \quad 1 \leq \hat{\theta}_{n,m}^{\textrm{MAD}} \leq \frac{n}{m}.
\end{equation*}
If we suppose that $d > n/k$ or $d > n/m$, then both estimators cannot expect to retrieve dependencies above the threholds stated by our two lemmas. In particular, in high dimension, i.e., when $d > n$, these estimators are unable to detect asymptotic independence. Indeed, in cases where the $d$ variables are asymptotically independent, extremes occur in one variable without influencing extremes in the others. Thus, when $d > n$, it is highly probable to observe a rank that is greater than $n-k+0.5$ for at least one variable and this occurs for every observation $i =1,\dots,n$. Since this happens for every observation $i =1,\dots,n$, the characteristic function is (with high probability) always equal to one, resulting in an overall extremal coefficient equal to $n/k$ when taking the sum. However, in high dimensions, this cannot be equal to $d$, the value taken in asymptotic independence. 

To illustrate these findings, we consider the following numerical setup. Consider as the sample size $n \in \{100,150, \dots, 1000\}$ and the high dimensional setting given by $d = n^{1.25}$ where we want to estimate the extremal coefficient of the random vector of $\textbf{Z}$ where its components are asymptotically independent. In this setup, we know that the theoretical value of $\theta$ is given by $d$. When studying the dependence structure of extreme events, we face the "curse of dimensionality" in two ways. Firstly, traditional estimators do not cover the full spectrum of possible values. Secondly, in order to expand the range of values, one may need to reduce the number of extremes considered, denoted by $k$, or decrease the size of block maxima, denoted by $m$. However, this may lead to an increase in variance or a decrease in bias, depending on the estimator used.

\begin{figure}[!h]

\begin{minipage}{.5\linewidth}
\centering
\subfloat[]{\label{subfig:eks}\includegraphics[scale=.5]{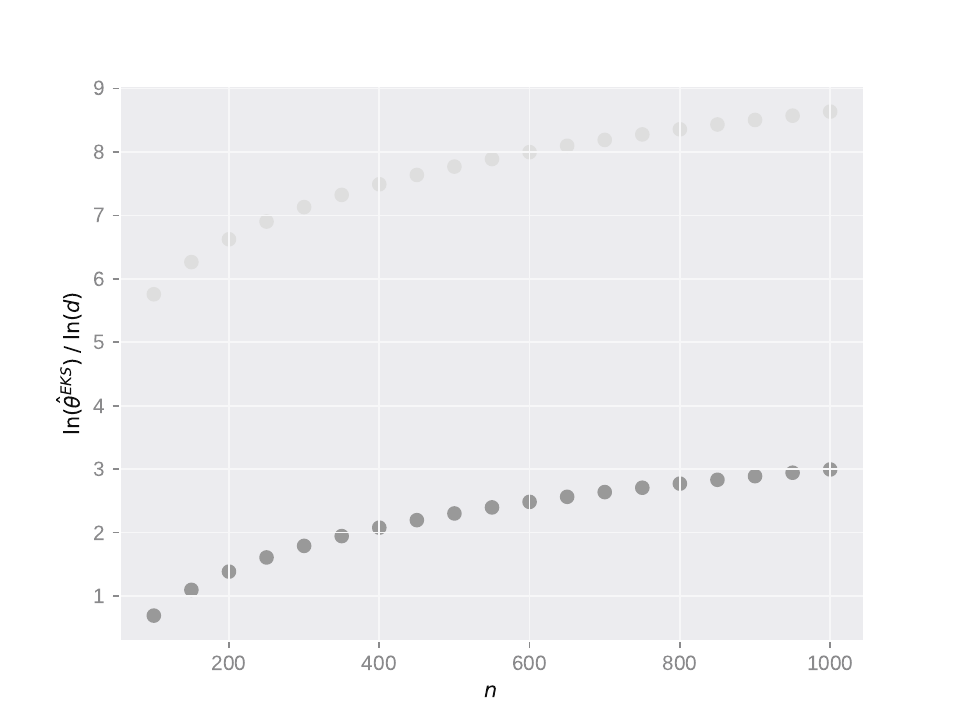}}
\end{minipage}%
\begin{minipage}{.5\linewidth}
\centering
\subfloat[]{\label{subfig:mado}\includegraphics[scale=.5]{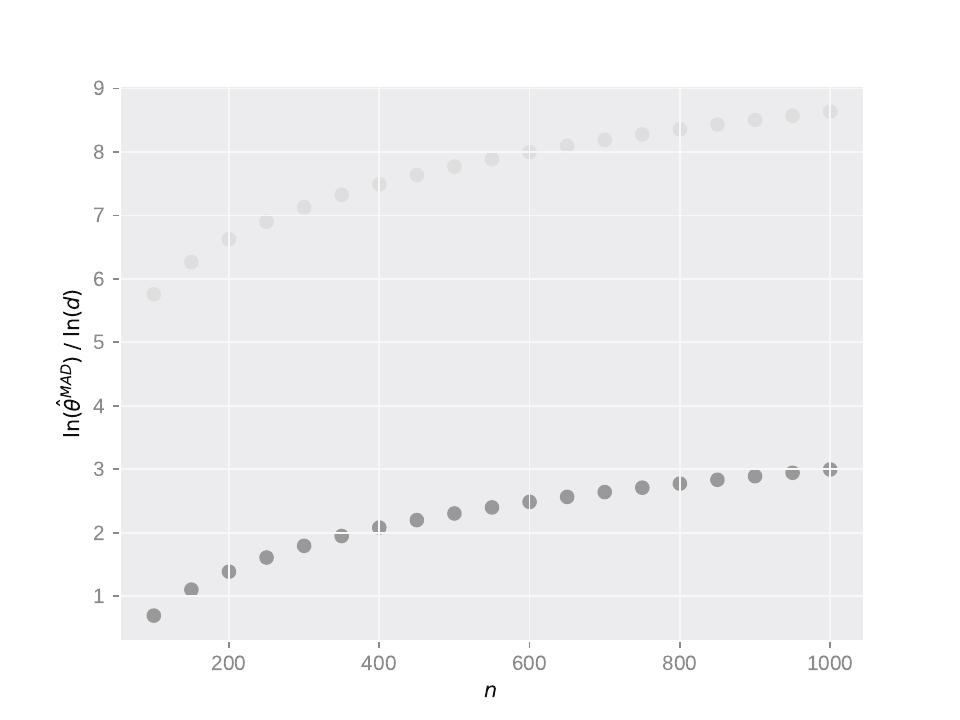}}
\end{minipage}

\caption{Estimator $\hat{\theta}_{n,k}^{\textrm{EKS}}$ (Panel \ref{subfig:eks}, black dots) and $\hat{\theta}_{n,m}^{\textrm{MAD}}$ (Panel \ref{subfig:mado}, black dots) according to different values of $n \in \{100,150,\dots,1000\}$ and $d = n^{1.25}$. The theoretical value of the extremal coefficient is depicted in grey dots. For both estimators, we took $k = m = 50$.}
\label{fig:app_high_dim}
\end{figure}

In Fig. \ref{fig:app_high_dim}, the bias of the estimator is clearly depicted. For each value of $n$, both estimators reach their upper bounds, that is $\hat{\theta}_{n,k}^{\textrm{EKS}} = n / k$ and $\hat{\theta}_{n,m}^{\textrm{MAD}} = n / m$.

\section{Consistent estimation of $\SECO$}
\label{proof:prop_consistency}

Consider a $q$-dimensional random vector $\textbf{Z} = (Z^{(1,1)},\dots,Z^{(d,p_d)})$, where $p_1+\dots+p_d = q$. This vector has a joint cumulative distribution function (c.d.f.) denoted as $F$, and each of its components has continuous marginal c.d.f.s $F^{(1,1)},\dots,F^{(d,p_d)}$. The copula, denoted as $C$, associated with $F$ (or equivalently, with $\textbf{Z}$) is defined as the c.d.f. of another random vector $\textbf{U} = (U^{(1,1)},\dots,U^{(d,p_d)})$. These random variables $\textbf{U}$ are obtained through the marginal application of the probability integral transform, meaning that $U^{(j,\ell)} = F^{(j,\ell)}(Z^{(j,\ell)})$ for $j = 1,\dots, d$ and $\ell = p_1,\dots,p_d$. Notably, the marginal c.d.f.s of the copula $C$ are uniformly distributed on the interval $[0,1]$. According to Sklar's theorem, the copula $C$ is a unique function that satisfies the following relationship for all $\textbf{x} = (x^{(1,1)},\dots,x^{(d,p_d)}) \in \mathbb{R}^q$:
\begin{equation*}
    F(x^{(1,1)},\dots,x^{(d,p_d)}) = C\left\{ F^{(1,1)}(x^{(1,1)}),\dots,F^{(d,p_d)}(x^{(d,p_d)}) \right\}.
\end{equation*}
In simpler terms, this equation describes how the joint distribution $F$ can be represented in terms of the copula $C$ and the marginal distributions $F^{(j,\ell)}$.

Now, let us consider a sequence of observed data, $\textbf{Z}_i$ for $i = 1,\dots,n$, which represents a stationary time series. Importantly, each $\textbf{Z}_i$ follows the same distribution as $\textbf{Z}$. Set $\textbf{U}_i = (U_i^{(1,1)},\dots,U_i^{(d,p_d)}) \sim C$ with $U_i^{(j,\ell)} = F^{(j,\ell)}(Z_i^{(j,\ell)})$. Define
\begin{equation*}
    \alpha_n^{(j)}(\textbf{u}) = \sqrt{n} \left( G_n^{(j)}(u) - u \right), \quad G_n^{(j)}(u) = n^{-1} \sum_{i=1}^n \mathds{1}_{ \{ U_i^{(j)} \leq u\} },
\end{equation*}
denote the (unobservable) empirical processes based on $\textbf{U}_1,\dots, \textbf{U}_n$.

For any sequence $(\textbf{Z}_n, n \in \mathbb{N})$, let

\begin{equation*}
    \mathcal{F}_k = \sigma(\textbf{Z}_n, n \leq k), \quad \textrm{ and } \mathcal{G}_k = \sigma(\textbf{Z}_n, n \geq k),
\end{equation*}
be the natural filtration and "reverse" filtration of the sequence $(\textbf{Z}_n, n \in \mathbb{N})$. Define

\begin{equation*}
    \beta(\mathcal{A}_1, \mathcal{A}_2) = \sup \, \frac{1}{2} \sum_{i,j \in I \times J} | \mathbb{P}(A_i \cap B_j) - \mathbb{P}(A_i) \mathbb{P}(B_j)|,
\end{equation*}

where the sup is taken over all finite $(A_i)_{i \in I}$ and $(B_j)_{j \in J}$ of $\Omega$ with the sets $A_i$ (resp. $B_j$) in the sigma field $\mathcal{A}_1$ (resp. $\mathcal{A}_2$). The $\beta$-mixing (or completely regular) coefficient is defined as

\begin{equation}
    \label{eq:beta_mixing}
    \beta(\ell) = \underset{n \in \mathbb{N}}{\sup} \, \beta(\mathcal{F}_n, \mathcal{G}_{n + \ell}).
\end{equation}

For the formulation of the consistency result for our estimator of the $\SECO$, we need a couple of conditions over the regularity of the sequence $(\textbf{Z}_n, n\in \mathbb{N})$ which are the following:

\begin{Assumption}{$\mathcal{A}$}
	    \label{ass:consistency}
	   	There exists an intermediary sequence $m = m_n$ such that $m_n = o (n)$ and $\beta(m_n) \rightarrow 0$ as $n \rightarrow \infty$, where $\beta$ is defined in \eqref{eq:beta_mixing}.
\end{Assumption}
\begin{Assumption}{$\mathcal{B}$}
	    \label{ass:regularity}
	   	There exists some $\theta_1 \in (0,1/2]$ such that, for all $\mu \in (0,\theta_1]$ and all sequences $\delta_n \rightarrow 0$, we have
     \begin{equation*}
         M_n(\delta_n, \mu) := \underset{|u-v| \leq \delta_n}{\sup} \, \frac{|\alpha_n^{(j)}(u) - \alpha_n^{(j)}(v)|}{\max\{|u-v|^\mu, n^{-\mu} \} } = o_{\mathbb{P}}(1), \quad j = 1,\dots,d
     \end{equation*}
\end{Assumption}
Condition \ref{ass:regularity} can, for instance, be verified in the i.i.d. case with $\theta_1 =1/2$ or for $\beta$-mixing sequence with $\beta(n) = o(a^n)$ as $n \rightarrow \infty$ for some $a\in (0,1)$, see Proposition 4.4 of \cite{berghaus2017weak}.
Here we state the proposition that the actual appendix is devoted to prove.

\begin{proposition}
    \label{prop:consistency}
    Let $k = k_n$ be an intermediary sequence. Provided that $k \rightarrow \infty$, $k /n \rightarrow 0$ as $n \rightarrow \infty$, under the regularity condition over the sequence $(\textbf{Z}_n, n \in \mathbb{N})$ stated by Condition \ref{ass:consistency} and Condition \ref{ass:regularity} and $\textbf{Z}$ is in the max-domain of attraction of $H$, then
    \begin{equation*}
        \widehat{\SECO}(\textbf{Z}^{(a)}, \textbf{Z}^{(b)}) \overunderset{\mathbb{P}}{ n \rightarrow \infty }{\longrightarrow} \SECO(\textbf{Z}^{(a)}, \textbf{Z}^{(b)}).
    \end{equation*}
\end{proposition}

\begin{proof} Since convergence in probability remains stable through addition and subtraction, our task is to establish the consistency of each component: $\hat{\theta}(a)$, $\hat{\theta}(b)$, and $\hat{\theta}(a,b)$. We address this concern in Lemma \ref{lem:lemma_as} below, where, without loss of generality, we focus on demonstrating the consistency of $\hat{\theta}(a)$. Also, without loss of generality, we study the following estimator

\begin{equation}
    \label{eq:ext_coeff_uniform}
    \hat{\theta}(a) = \frac{1}{k} \sum_{i=1}^n \mathds{1}_{ \{ \hat{F}_n^{(a,1)}(Z_i^{(a,1)}) > 1-k/n \textrm{ or } \dots \textrm{ or } \hat{F}_n^{(a,p)}(Z_i^{(a,p)}) > 1-k/n \} },
\end{equation}
where $\hat{F}_n^{(a,\ell)}$ denotes the empirical distribution function of $Z_1^{(a,\ell)},\dots, Z_n^{(a,\ell)}$ for $\ell = 1,\dots,p_a$. This estimator mirrors the one presented in \eqref{eq:emp_eco}, with the exception that we employ uniform margins instead of ranks, and we omit the constant factor of $1/2$, which, crucially, does not alter the estimator's asymptotic behavior.

\begin{lemma}
    \label{lem:lemma_as}
    Under the conditions of Proposition \ref{prop:consistency}, we have
    \begin{equation*}
        \hat{\theta}(a) = \theta_q(a) + o_{\mathbb{P}}(1),
    \end{equation*}
    with
    \begin{equation*}
        \theta_t(a) = \mathbb{P}\left\{ F^{(a,1)}(Z^{(a,1)}) > 1-t \textrm{ or } \dots \textrm{ or } F^{(a,p)}(Z^{(a,p_a)}) > 1-t \right\} / t,
    \end{equation*}
    and $\underset{t \rightarrow 0}{\lim} \, \theta_t(a) = \theta(a)$.
\end{lemma}

\begin{proof}
    Without confusions, we set in this proof $p_a = p$ and that $\textbf{Z} = (Z^{(1)},\dots,Z^{(p)}) := (Z^{(a,1)},\dots,Z^{(a,p)})$ is a $p$-dimensional random vector. We begin by introducting some useful notations. In the same spirit, define the random variable $U_i^{(j)} = F^{(j)}(Z_i^{(j)})$ (here $Z_{i}^{(j)}$ denote the $j$th entry of the vector $\textbf{Z}_i$) and the vectors $\textbf{U}_i := (U_i^{(1)},\dots,U_i^{(p)})$ with stationary distribution $C$. Denote by $\hat{G}_n^{(j)}$ the empirical distribution of $U_1^{(j)},\dots,U_n^{(j)}$. Define the vector $\hat{G}_n^\leftarrow(\textbf{x}) = ((\hat{G}_n^{(1)})^\leftarrow(x^{(1)}), \dots, (\hat{G}_n^{(p)})^\leftarrow(x^{(p)}))$, the function
    \begin{equation*}
        C_n^{o}(\textbf{x}) = \frac{1}{n} \sum_{i=1}^n \mathds{1}_{\{ U_i^{(1)} \leq x^{(1)}, \dots, U_i^{(p)} \leq x^{(p)} \} }
    \end{equation*}
    and
    \begin{equation*}
        \hat{C}_n \left( 1-\frac{k \textbf{x}}{n} \right) := C_n^{o}\left( \hat{G}_n^\leftarrow \left(1-\frac{k \textbf{x}}{n}\right) \right)
    \end{equation*}
    Note that the estimator $\hat{\theta}$ depends only on the marginals ranks of $Z_i^{(j)}$ with $j =1,\dots,p$; thus we have almost surely
    \begin{equation*}
        |\hat{\theta}(a) - \theta_t(a) | = \left| \frac{n}{k} \hat{C}_n(1-k/n, \dots, 1-k/n) - t^{-1} C(1-t, \dots, 1-t) \right|
    \end{equation*}
    Standard arguments gives that under $k \rightarrow \infty$ and $k /n \rightarrow t \in (0,1)$ and $n \rightarrow \infty$, the right hand side of the latter equation is equal to
    \begin{equation*}
        \left| \frac{n}{k}\left(\hat{C}_n\left(1-\frac{k}{n}, \dots, 1-\frac{k}{n}\right) - C\left(1-\frac{k}{n}, \dots, 1-\frac{k}{n}\right) \right) \right| + o_{\mathbb{P}}(1).
    \end{equation*}
    Now, we can bound the first term by
    \begin{equation*}
        \frac{n}{k} \left| C_n^{o}(\hat{G}_n^\leftarrow(1-k\textbf{1}/n)) - C(\hat{G}_n^\leftarrow(1-k\textbf{1}/n)) \right| + \frac{n}{k} \left| C(\hat{G}_n^\leftarrow(1-k\textbf{1}/n)) - C(1-k\textbf{1}/n) \right|.
    \end{equation*}
    Using Lipschitz continuity of $C$, we obtain the following upper bound:
    \begin{equation}
        \label{eq:main_prop}
        \frac{n}{k} || C_n^{o} - C ||_\infty + \frac{n}{k} \sum_{j=1}^d || u_n^{(j)} ||_\infty
    \end{equation}
    where $|| \cdot ||_\infty$ is the uniform norm and $u_n^{(j)}(u) = (\hat{G}_n^{(j)})^\leftarrow(u) - u$, $u \in [0,1]$.

    By Berbee's coupling Lemma (\cite{AIHPB_1995__31_2_393_0, bucher2014extreme}), one can construct inductively a sequence $(\bar{\textbf{Z}}_{im+1},\dots, \bar{\textbf{Z}}_{im+m})_{i \geq 0}$ such that the following three properties hold:
	\begin{enumerate}[label=(\roman*)]
		\item $(\bar{\textbf{Z}}_{im+1},\dots, \bar{\textbf{Z}}_{im+m}) \overset{d}{=} (\textbf{Z}_{im+1},\dots, \textbf{Z}_{im+m})$ for any $i \geq 0$; \label{property:law}
		\item both $(\bar{\textbf{Z}}_{2im+1},\dots, \bar{\textbf{Z}}_{2im+m})_{i\geq 0}$ and $(\bar{\textbf{Z}}_{(2i+1)m+1},\dots, \bar{\textbf{Z}}_{(2i+1)m+m})_{i\geq 0}$ sequences are independent and identically distributed; \label{property:iid}
		\item $\mathbb{P}\{(\bar{\textbf{Z}}_{im+1},\dots, \bar{\textbf{Z}}_{im+m}) \neq (\textbf{Z}_{im+1}, \dots, \textbf{Z}_{im+m}) \} \leq \beta(m)$. \label{property:mixing}
	\end{enumerate}
    Let $\bar{C}_n^o$ be defined analogously to $C_n^o$ but with $\textbf{Z}_1,\dots,\textbf{Z}_n$ replaced by $\bar{\textbf{Z}}_1,\dots,\bar{\textbf{Z}}_n$. Now write
    \begin{equation}
    \label{eq:summand}
        C_n^o(\textbf{u}) - C(\textbf{u}) = \left\{ C_n^o(\textbf{u}) - \bar{C}_n^o(\textbf{u}) \right\} + o_{\mathbb{P}}(1).
    \end{equation}
    The term in brackets in the right hand side is $o_\mathbb{P}(1)$ uniformly in $\textbf{u}$, since
    \begin{equation*}
        |C_n^o(\textbf{u}) - \bar{C}_n^o(\textbf{u})| \leq \frac{1}{n} \sum_{i=1}^n \mathds{1}_{\{\bar{\textbf{Z}}_i \neq \textbf{Z}_i \} }.
    \end{equation*}
    Hence by Markov's inequality, for any $\epsilon > 0$
    \begin{equation*}
        \mathbb{P}\left\{ \underset{\textbf{u} \in [0,1]^q}{\sup} \left| C_n^o(\textbf{u}) - \bar{C}_n(\textbf{u}) \right| > \epsilon \right\} \leq \frac{\beta(m)}{\epsilon}.
    \end{equation*}
    By Condition \ref{ass:consistency}, we obtain that the first summand in brackets in \eqref{eq:summand} is $o_{\mathbb{P}}(1)$ as $n\rightarrow \infty$, uniformly in $\textbf{u} \in [0,1]^q$. We obtain that
    \begin{equation*}
        || C_n^{o} - C ||_\infty = o_\mathbb{P}(1).
    \end{equation*}
    We also have
    \begin{equation*}
        \sup_{u \in [0,1]} \, |u_n^{(j)}(u) | \leq \sup_{u \in [0,1]} \, |\alpha_n^{(j)}(u)| + \frac{n \sup_{u\in [0,1]} |\hat{G}_n^{(j)}(u) - \hat{G}_n^{(j)}(u-)|-1}{n}.
    \end{equation*}
    To understand this, we start by recognizing that the maximum of either $\alpha_n^{(j)}(\cdot)$ or $-\alpha_n^{(j)}(\cdot)$, and consequently, $|\alpha_n^{(j)}(\cdot)|$, must occur at one of the discontinuities in $\hat{G}_n^{(j)}.$ These discontinuities correspond to the values $\{ U_{i:n}^{(j)}, \; 1\leq i \leq n\}$, where $U_{1:n}^{(j)} \leq \dots \leq U_{n:n}^{(j)}$ represent the order statistics. Hence, the quantity
    \begin{equation}
        \label{eq:count}
        n \times (u_n^{(j)}(i/n) + \alpha_n^{(j)}(U_{i:n}^{(j)}))
    \end{equation}
    is equal to the highest count of $U_i^{(j)}$ that are equal to $U_{i:n}^{(j)}$ minus 1. Assuming there are no ties among $U_1^{(j)},\dots,U_n^{(j)}$ (which, for example, happens in the i.i.d. case), this expression equals 1. Consequently, we derive the classical identity for the uniform quantile process, as outlined in \cite[Section 1.4]{csorgHo1983quantile} or \cite[Chapter 3]{shorack2009empirical}.
    \begin{equation*}
        \sup_{0 \leq u \leq 1} \, |\alpha_n^{(j)}(u)| = \sup_{0 \leq u \leq 1} \, |u_n^{(j)}(u)|.
    \end{equation*}
    In the general case, Equation \eqref{eq:count} is limited above by the maximum count of $U_i^{(j)}$ which are equal minus 1. It is worth noting that this maximum count can be expressed as
    \begin{equation*}
        n \times \underset{u \in [0,1]}{\sup} \, |\hat{G}_n^{(j)}(u) - \hat{G}_n^{(j)}(u-)|.
    \end{equation*}
    
    We have, following the proof of lemma 4.6 in \cite{berghaus2017weak}
    \begin{align*}
        \sup_{u\in [0,1]} |\hat{G}_n^{(j)}(u) - \hat{G}_n^{(j)}(u-)| &\leq \sup_{u,v\in [0,1], |u-v| \leq 1/n} |\hat{G}_n^{(j)}(u) - \hat{G}_n^{(j)}(v)| \\
        &\leq \sup_{u,v\in [0,1], |u-v| \leq 1/n} |\hat{G}_n^{(j)}(u) - \hat{G}_n^{(j)}(v) - (u-v)| + \frac{1}{n} \\
        &\leq \frac{1}{\sqrt{n}} \sup_{u,v\in [0,1], |u-v| \leq 1/n} |\alpha_n^{(j)}(u) - \alpha_n^{(j)}(v)| + \frac{1}{n}.
    \end{align*}
    Using Condition \ref{ass:regularity}, the above term is $o_\mathbb{P}(n^{-1/2-\mu})$ for $\mu \in (0,\theta_1)$ and $\theta_1 \in [0,1/2]$. Additionally, using $||C_n^{o} - C||_\infty = o_\mathbb{P}(1)$, we obtain
    \begin{equation*}
        ||u_n^{(j)}||_\infty = o_{\mathbb{P}}(1).
    \end{equation*}
    Thus, $\forall t \in (0,1)$
    \begin{equation*}
        \hat{\theta}_n(a) = \theta_t(a) + o_{\mathbb{P}}(1).
    \end{equation*}
    Since $F$ is in the max-domain of attraction of $H$, we have
    \begin{equation*}
        \theta(a) = \underset{t \rightarrow 0}{\lim} \, \theta_t(a).
    \end{equation*}
    Hence the result of the lemma.
\end{proof}
Since convergence in probability is reliably preserved under continuous transformations, we attain the outcomes outlined in Proposition \ref{prop:consistency}.
\end{proof}

\section{Definition of the Adjusted Rand Index (ARI)}
\label{sec:ARI}
The ARI is computed as follows: Let $O = \{O_g\}_{g=1,\dots,G}$ and $S = \{S_h\}_{h=1,\dots,H}$ be two partitions with $d$ entities, and let $d_{gh}$ be the number of entities in cluster $O_g$ in partition $O$ and in cluster $S_h$ in partition $S$. Denote by $d_{g\boldsymbol{\cdot}}$ (resp. $d_{\boldsymbol{\cdot}h}$) the number of entities in cluster $O_g$ (resp. $S_h$) in partition $O$ (resp. $S$). The ARI is evaluated using the following expressions:
\begin{align*}
    &r_0 = \sum_{g=1}^G \sum_{h=1}^H \binom{d_{gh}}{2}, \quad r_1 = \sum_{g=1}^G \binom{d_{g\boldsymbol{\cdot}}}{2}, \quad r_2 = \sum_{g=1}^G \binom{d_{\boldsymbol{\cdot}h}}{2}, \quad r_3 = \frac{2 r_1 r_2}{d(d-1)},\\ &\textrm{ARI}(O,S) = \frac{r_0-r_3}{0.5(r_1+r_2)-r_3},
\end{align*}
where $\binom{n}{k}$ is the binomial coefficient.

\section{Supplementary Figures}
\label{sec:sup_fig}

\begin{figure}
    \begin{minipage}{.33\linewidth}
        \centering
        \subfloat[]{\includegraphics[scale=.30]{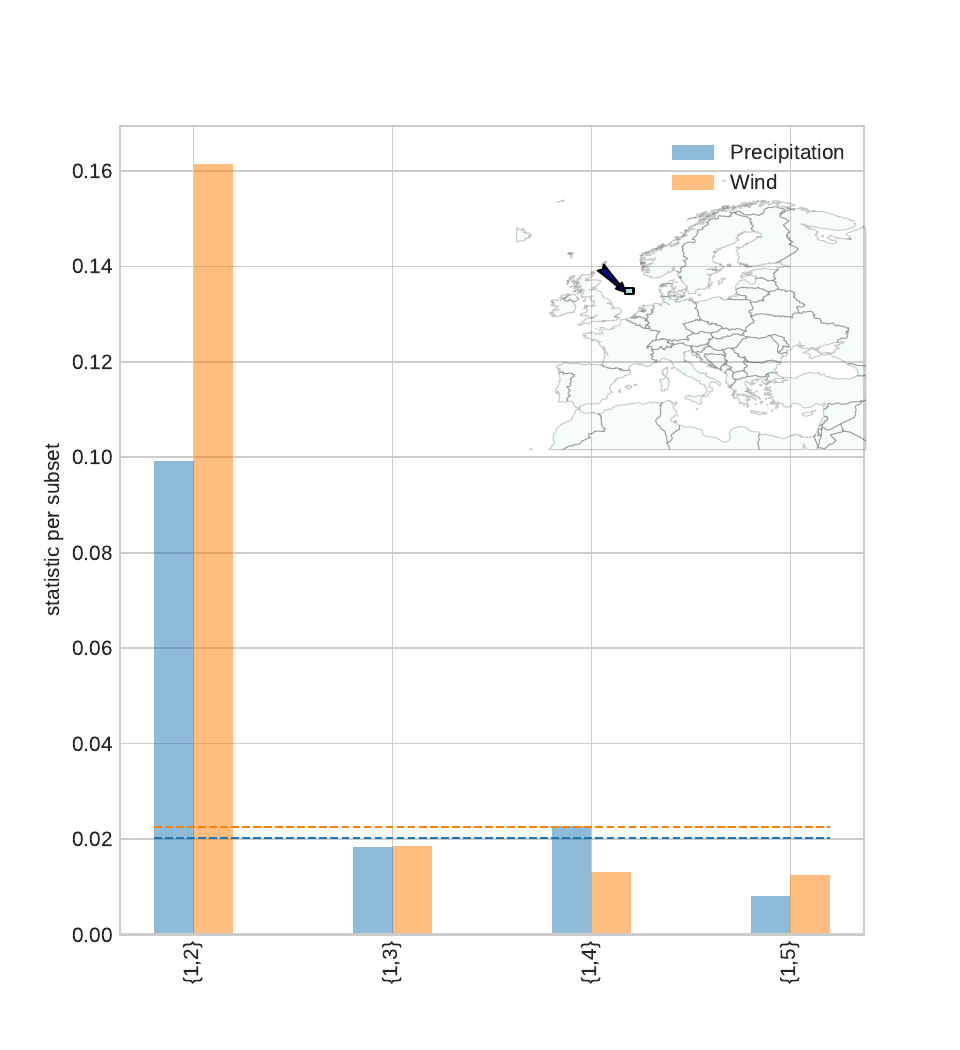}}
    \end{minipage}%
    \begin{minipage}{.33\linewidth}
        \centering
        \subfloat[]{\includegraphics[scale=.30]{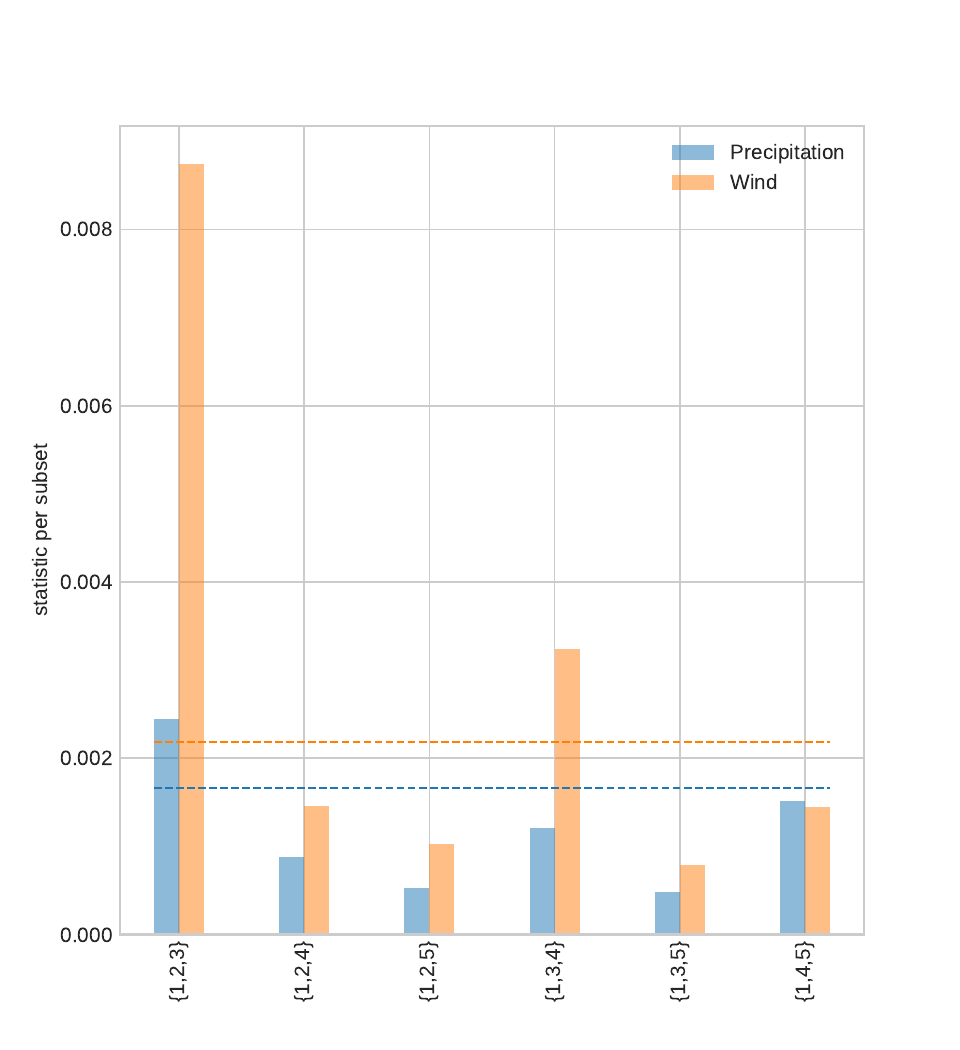}}
    \end{minipage}
    \begin{minipage}{.33\linewidth}
        \centering
        \subfloat[]{\includegraphics[scale=.30]{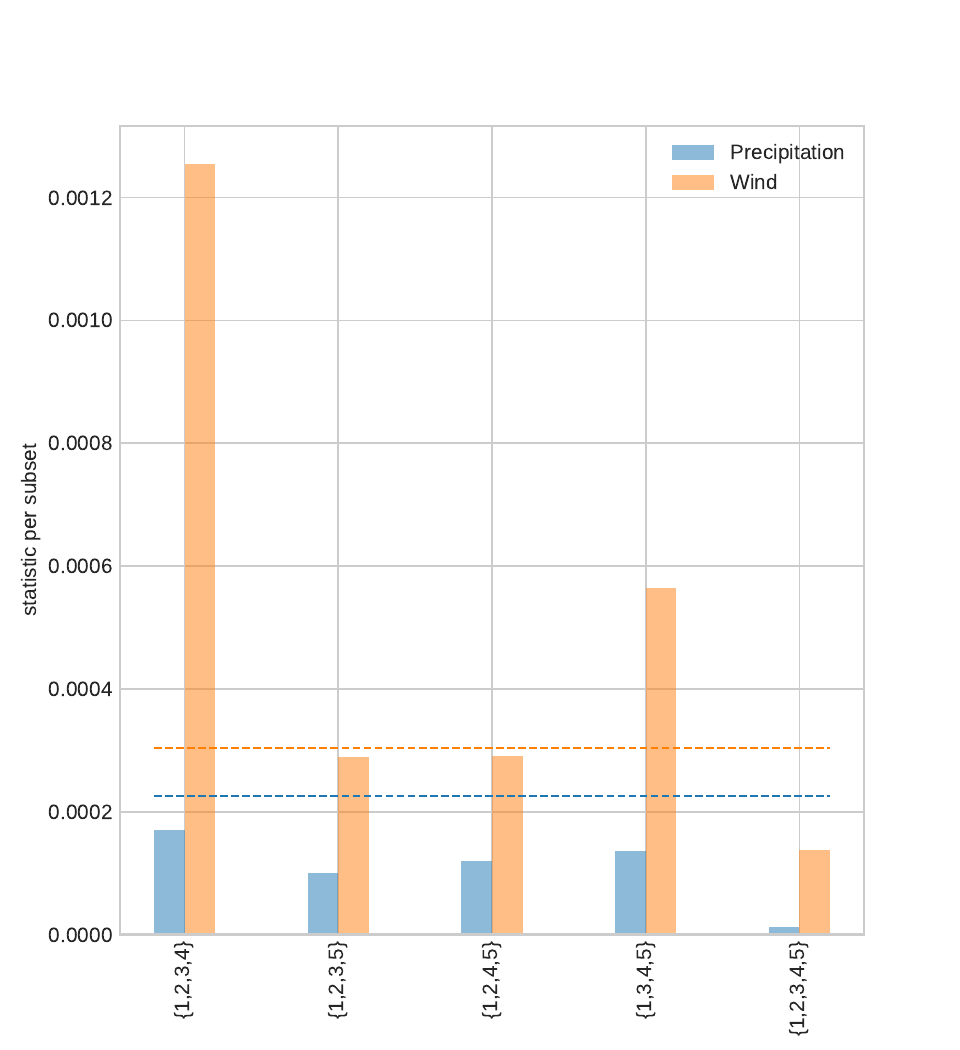}}
    \end{minipage}
    \begin{minipage}{.33\linewidth}
        \centering
        \subfloat[]{\includegraphics[scale=.30]{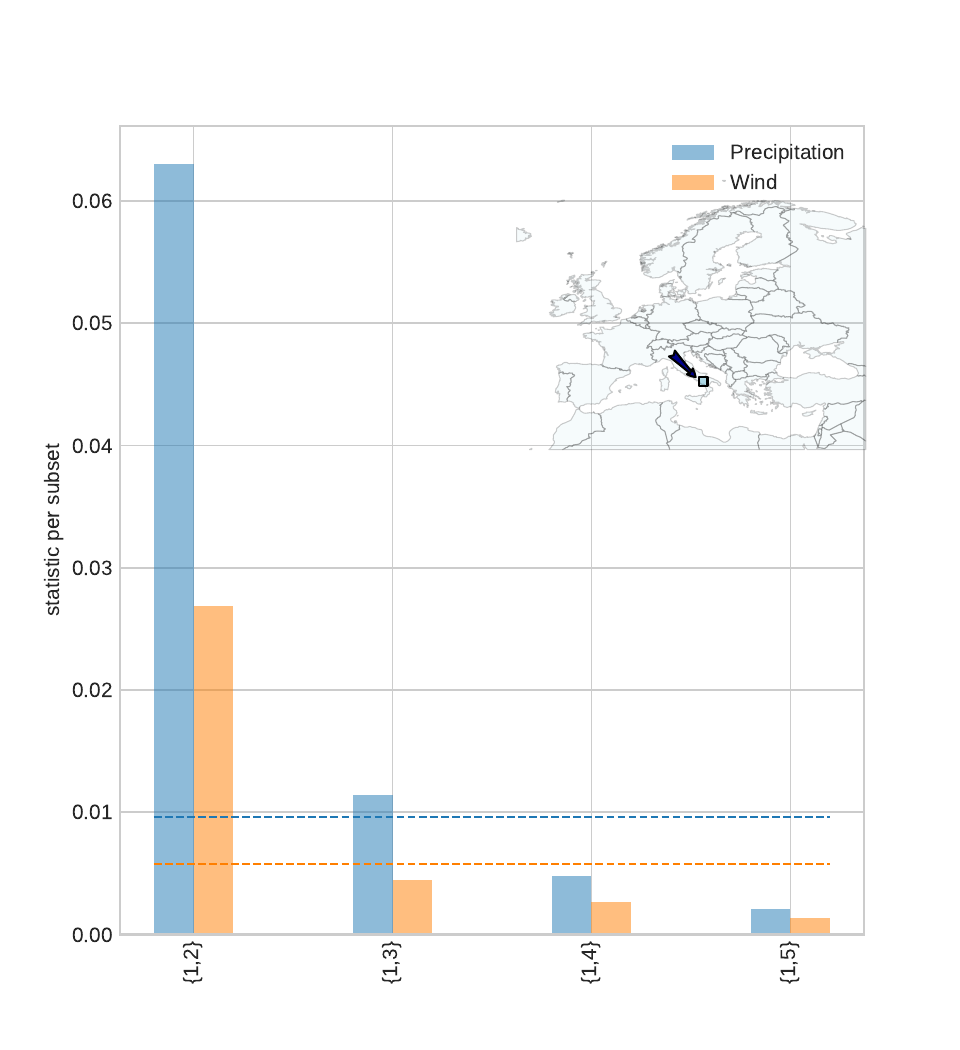}}
    \end{minipage}%
    \begin{minipage}{.33\linewidth}
        \centering
        \subfloat[]{\includegraphics[scale=.30]{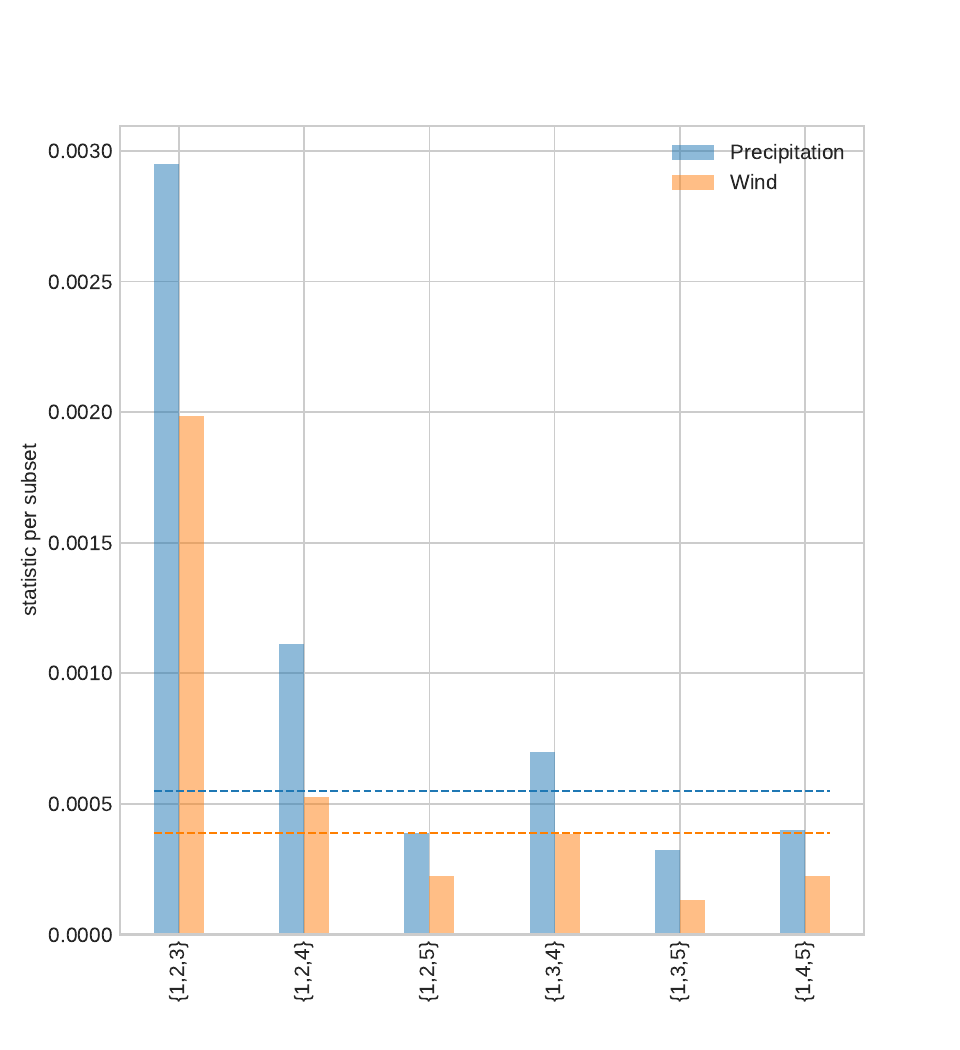}}
    \end{minipage}
    \begin{minipage}{.33\linewidth}
        \centering
        \subfloat[]{\includegraphics[scale=.30]{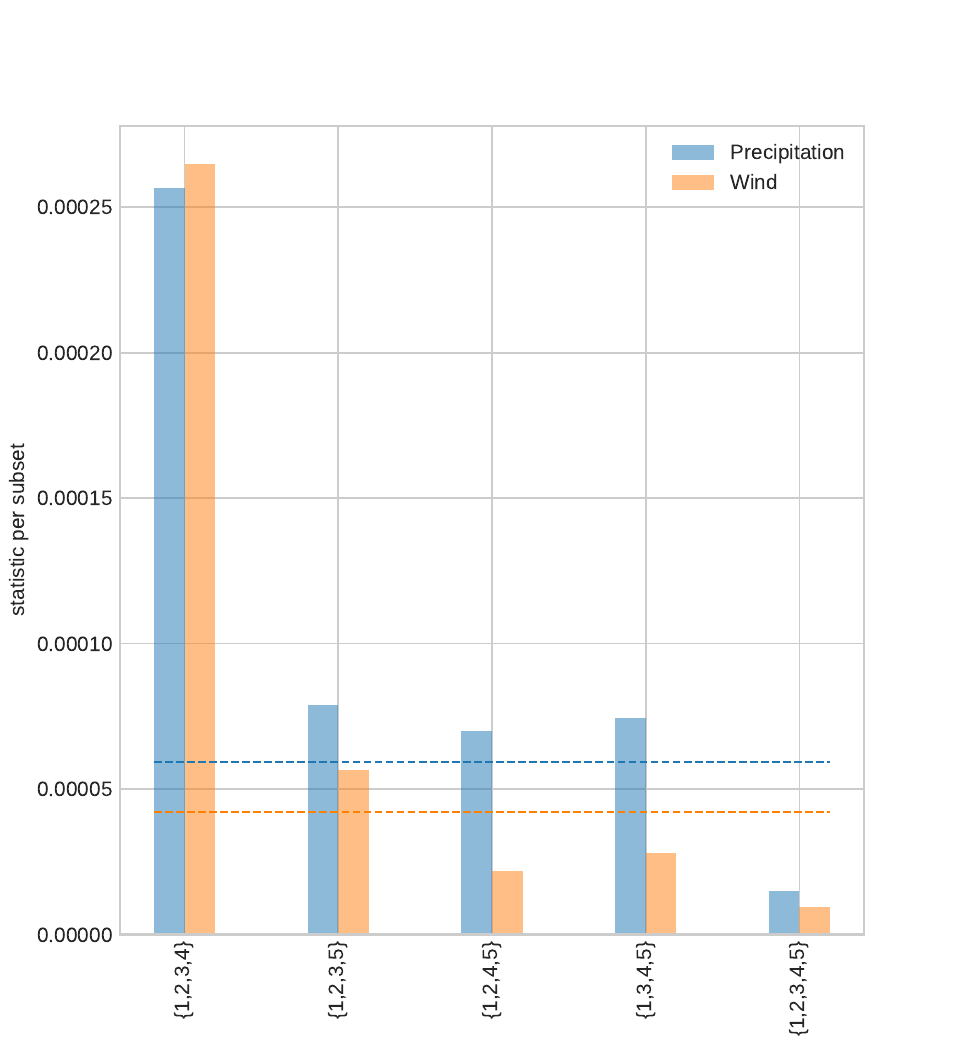}}
    \end{minipage}
    \begin{minipage}{.33\linewidth}
        \centering
        \subfloat[]{\includegraphics[scale=.30]{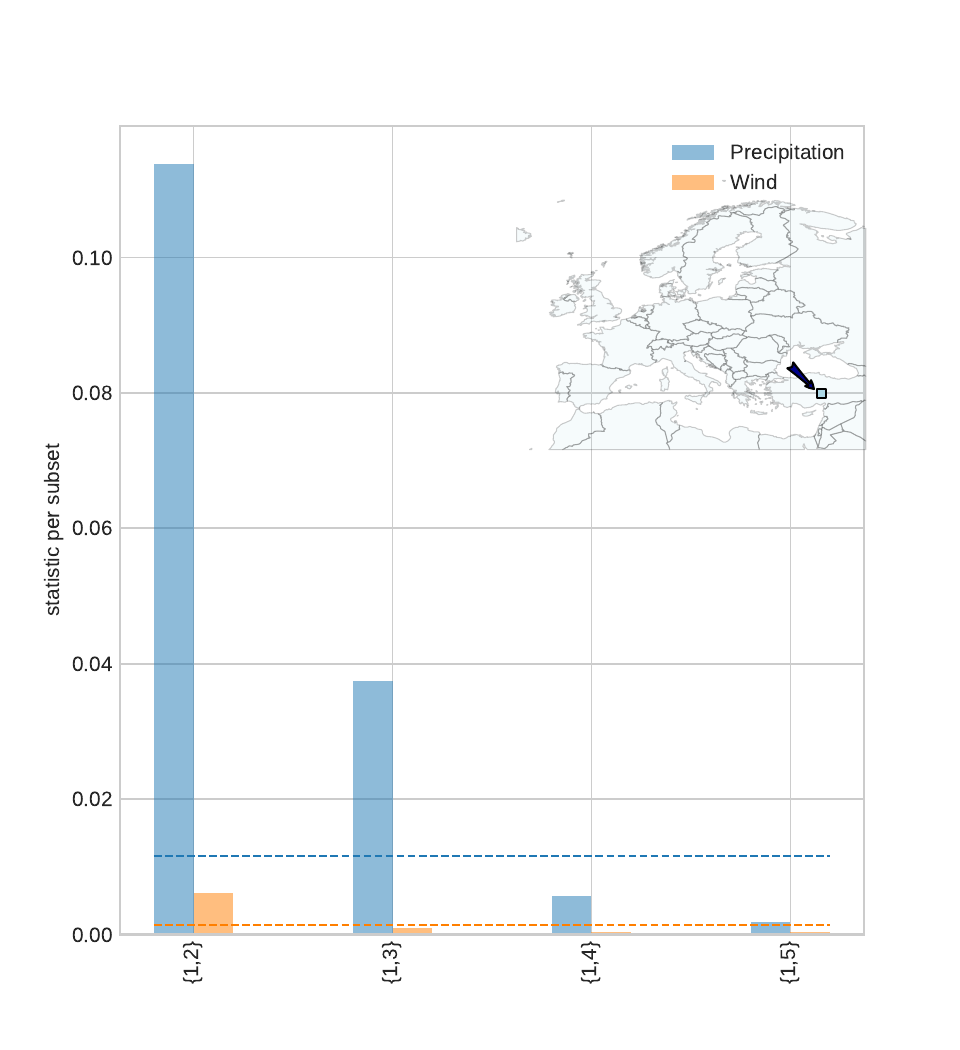}}
    \end{minipage}%
    \begin{minipage}{.33\linewidth}
        \centering
        \subfloat[]{\includegraphics[scale=.30]{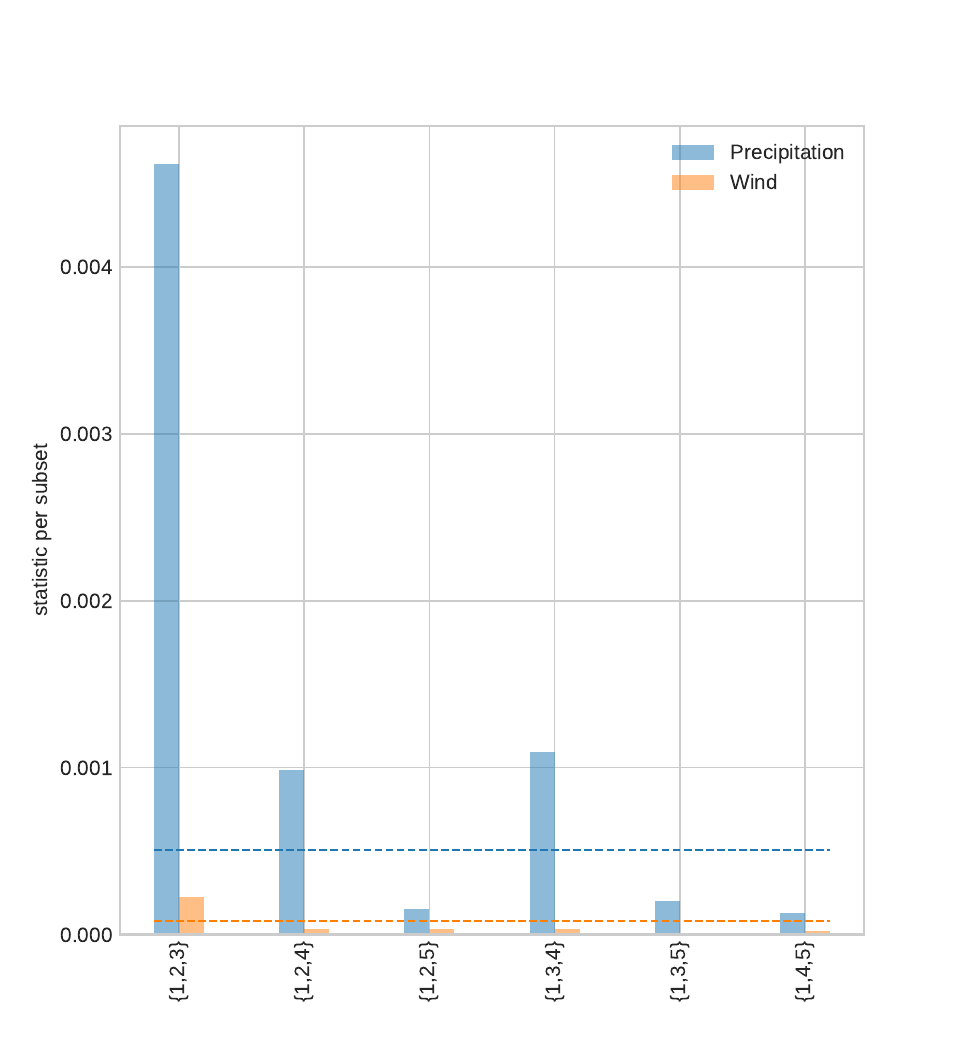}}
    \end{minipage}
    \begin{minipage}{.33\linewidth}
        \centering
        \subfloat[]{\includegraphics[scale=.35]{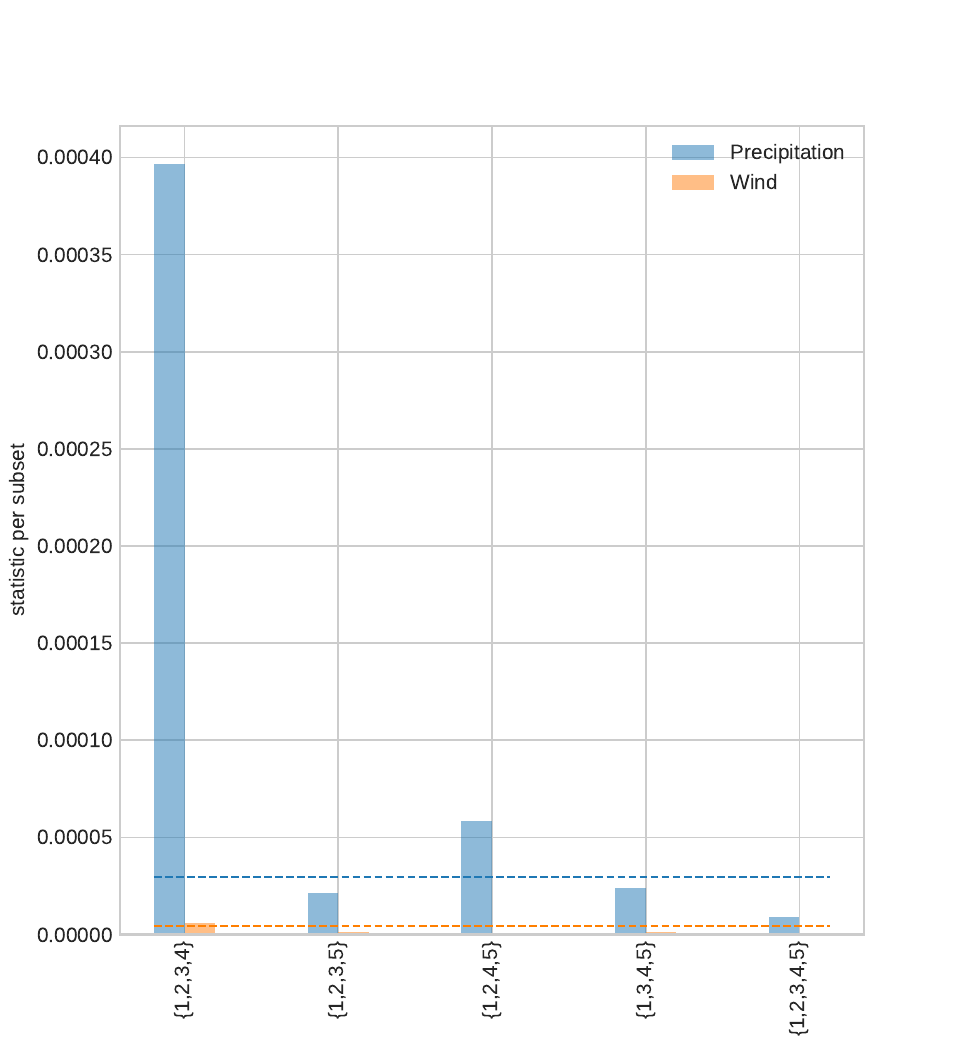}}
    \end{minipage}
    \caption{The dependogram provides a concise summary of randomness test results conducted on daily total precipitation and wind speed maxima in the ERA5 dataset. This study covers three European regions, examining 304 days of observations across nine distinct pixels. The first column displays a map with each row representing an area and the nine pixels marked by red squares. Test statistics are represented by bars, and critical values are depicted by dotted horizontal lines. The dependogram columns focus on pairwise, three-wise, and four-wise randomness tests with a lag of $4$.}
    \label{fig:dependogram}
\end{figure}

\begin{figure}[!htb]

\begin{minipage}{.5\linewidth}
\centering
\subfloat[]{\label{subfig:seco_tp}\includegraphics[scale=.4]{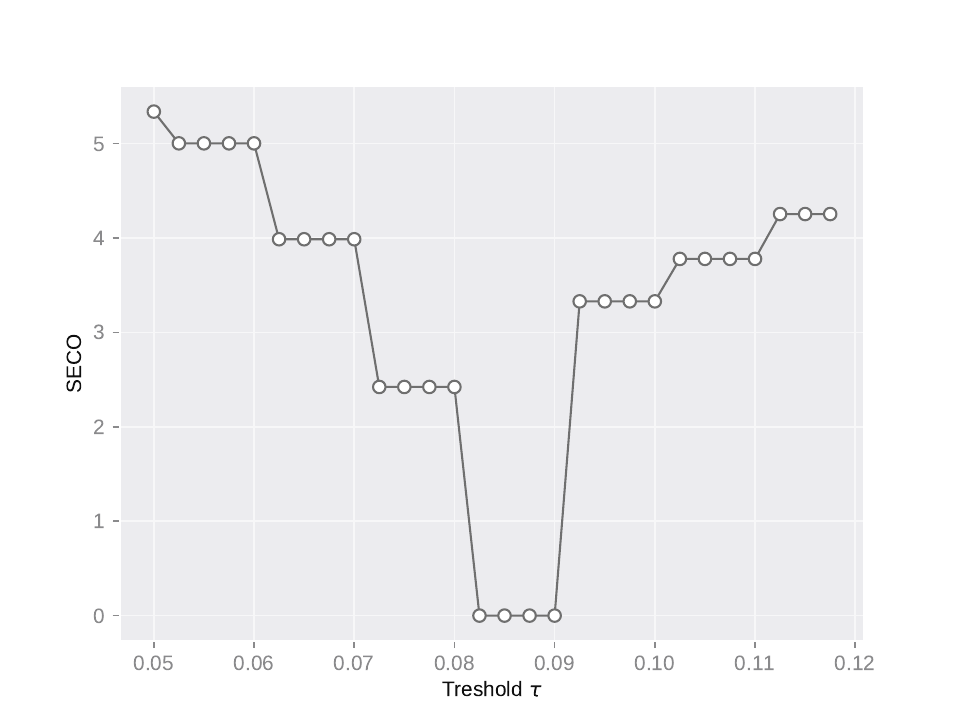}}
\end{minipage}%
\begin{minipage}{.5\linewidth}
\centering
\subfloat[]{\label{subfig:seco_wind}\includegraphics[scale=.4]{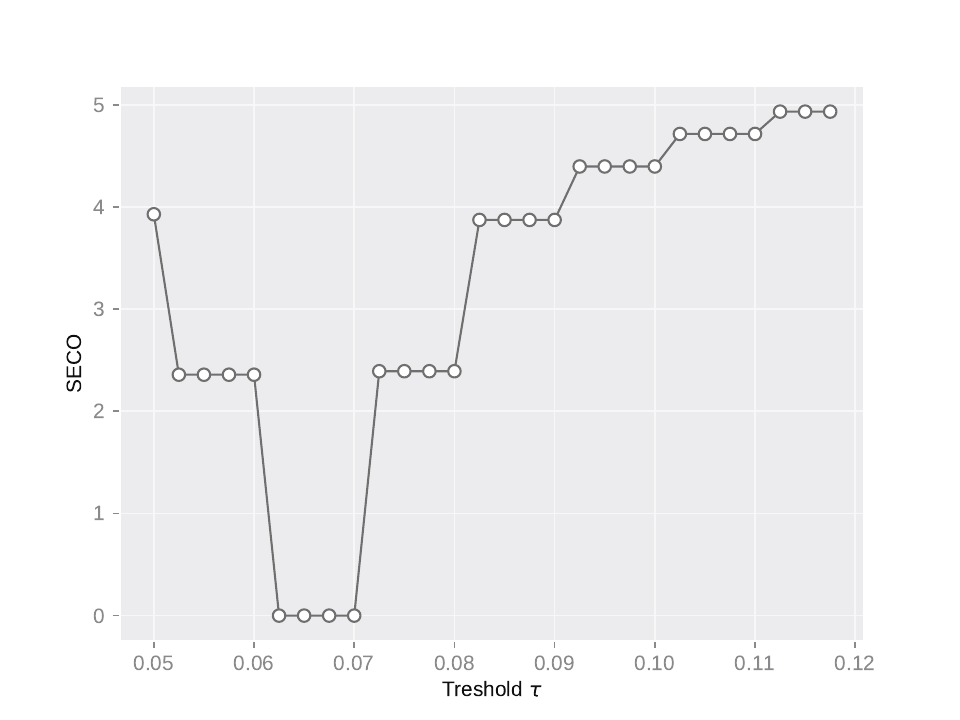}}
\end{minipage}
\begin{minipage}{.5\linewidth}
\centering
\subfloat[]{\label{subfig:matseco_tp}\includegraphics[scale=.45]{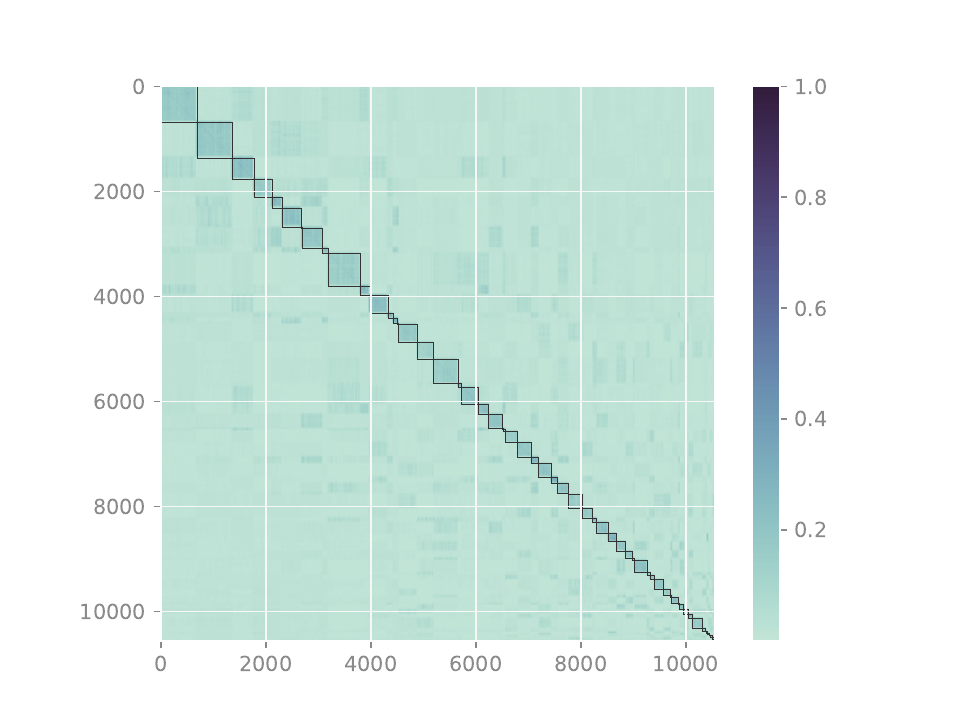}}
\end{minipage}%
\begin{minipage}{.5\linewidth}
\centering
\subfloat[]{\label{subfig:matseco_wind}\includegraphics[scale=.45]{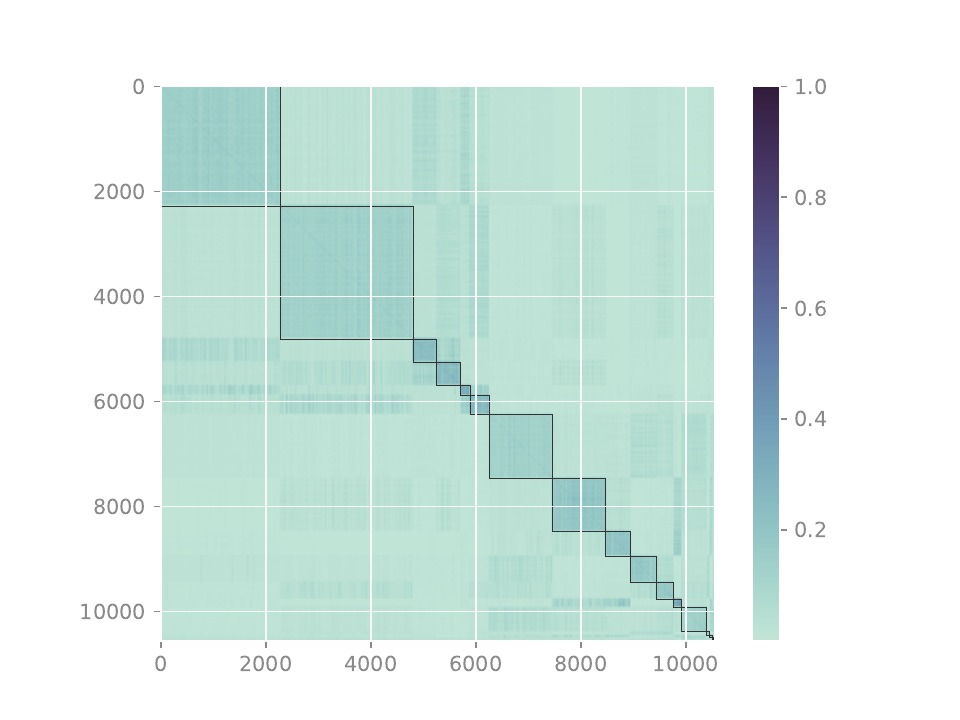}}
\end{minipage}
\begin{minipage}{.5\linewidth}
\centering
\subfloat[]{\label{subfig:clust_1_9_tp}\includegraphics[width=8cm, height=7cm]{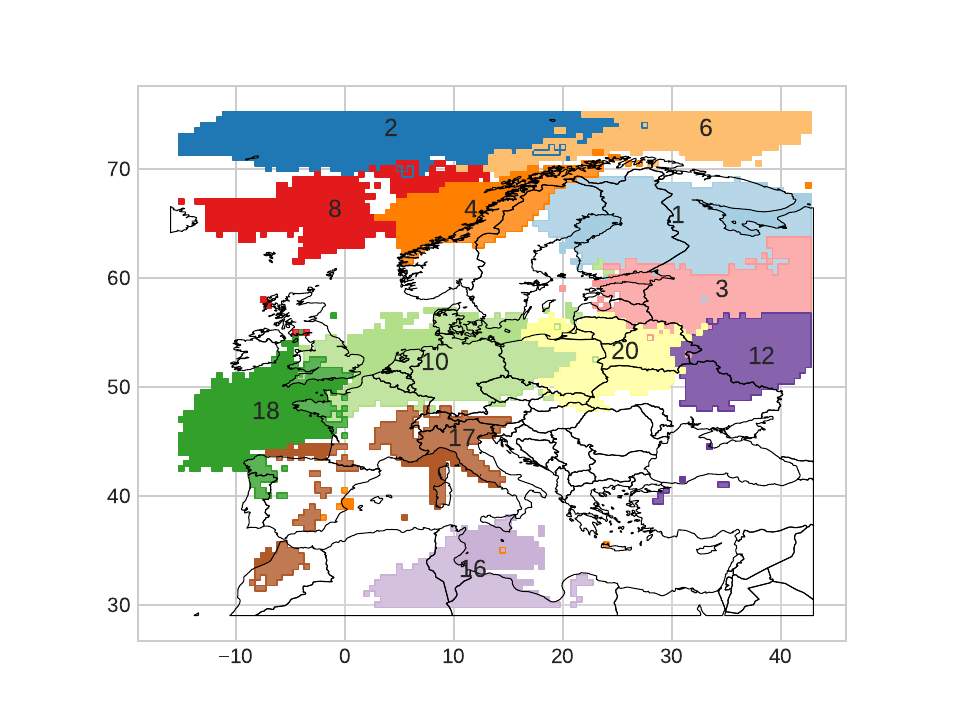}}
\end{minipage}%
\begin{minipage}{.5\linewidth}
\centering
\subfloat[]{\label{subfig:clust_1_9_wind}\includegraphics[width=8cm, height=7cm]{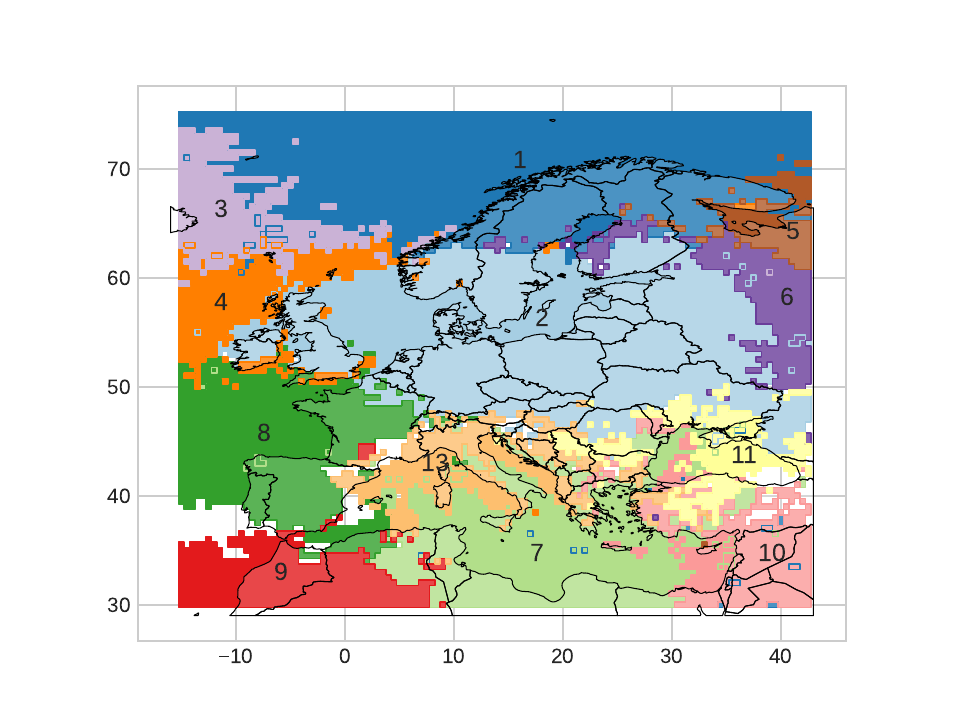}}
\end{minipage}

\caption{Value of the function $L$ for different values of $\tau \in \Delta = \{0.05,0.0525,\dots,0.12\}$ in Panels \ref{subfig:seco_tp} (Precipitation) and \ref{subfig:seco_wind} (Wind). Partitions of the extremal correlation similarity matrix with threshold $\tau = 0.09$ for Panel \ref{subfig:matseco_tp} (Precipitation) and $\tau = 0.07$ \ref{subfig:matseco_wind} (Wind). Squares represent the clusters of variables. Representation of the $12$ largest clusters (in decreasing order) of the partition of the extremal correlation matrix of total precipitation and wind speed maxima with threshold $\tau = 0.09$ and $\tau = 0.07$, respectively in Panels \ref{subfig:clust_1_9_tp} and \ref{subfig:clust_1_9_wind}. Number of each cluster is depicted at the top-left corner of the corresponding panel.}
\label{fig:result_clust_tp_wind}
\end{figure}

\begin{figure}[!h]
    \begin{minipage}{.5\linewidth}
    \centering
    \subfloat[]{\label{subfig:quanti_silhouette}\includegraphics[scale=.5]{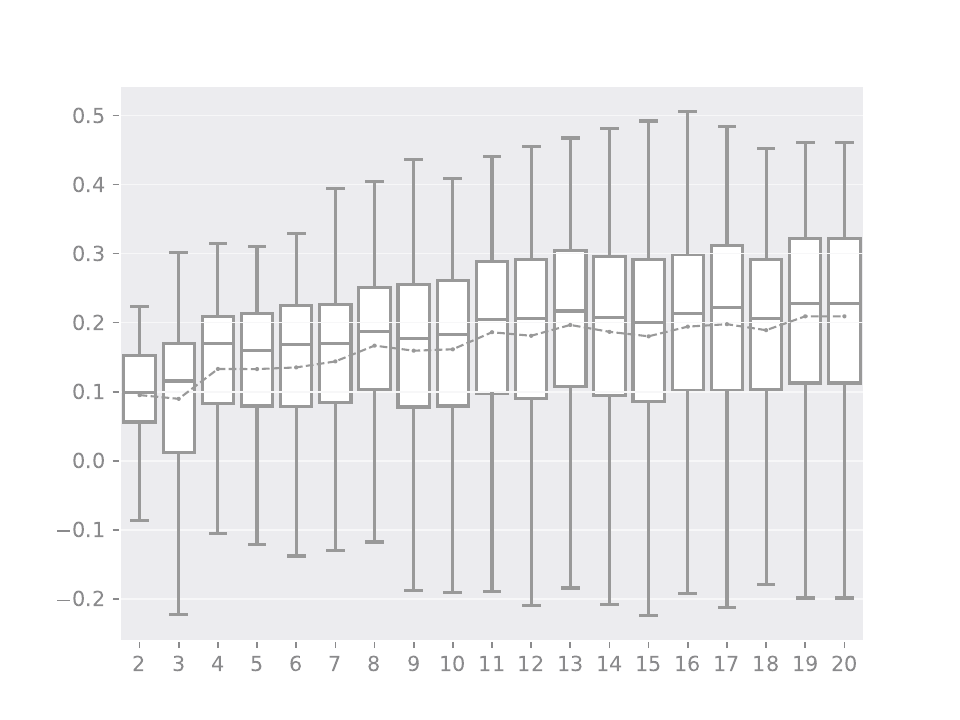}}
    \end{minipage}%
    \begin{minipage}{.5\linewidth}
    \centering
    \subfloat[]{\label{subfig:clust_compound_quanti}\includegraphics[scale=.5]{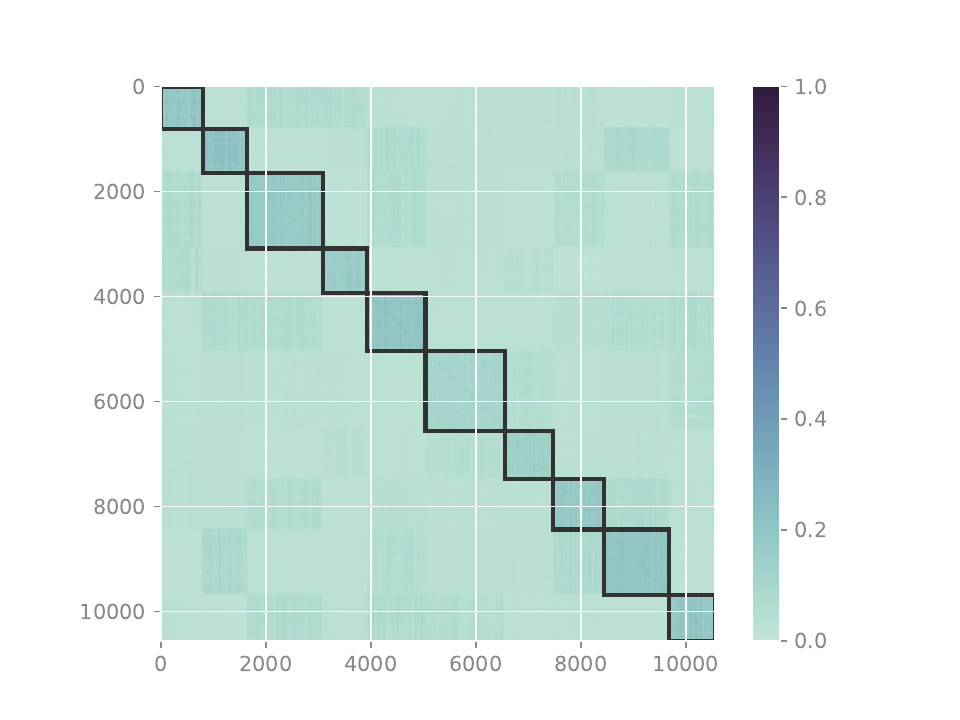}}
    \end{minipage}
    \caption{Boxplots of the silhouette coefficients for different values of K using the quantization-based algorithm. Thick lines indicate the median, boxes the interquartile range and whiskers the full range of the distribution. Partitions of the $\SECO$ similarity matrix with $K = 10$ for the quantization-based approach. Squares represent the clusters of variables.} 
    \label{fig:clust_quanti_matrix}
\end{figure}

\begin{figure}[!htb]

\begin{minipage}{.5\linewidth}
\centering
\subfloat[]{\label{subfig:silhouette}\includegraphics[scale=.5]{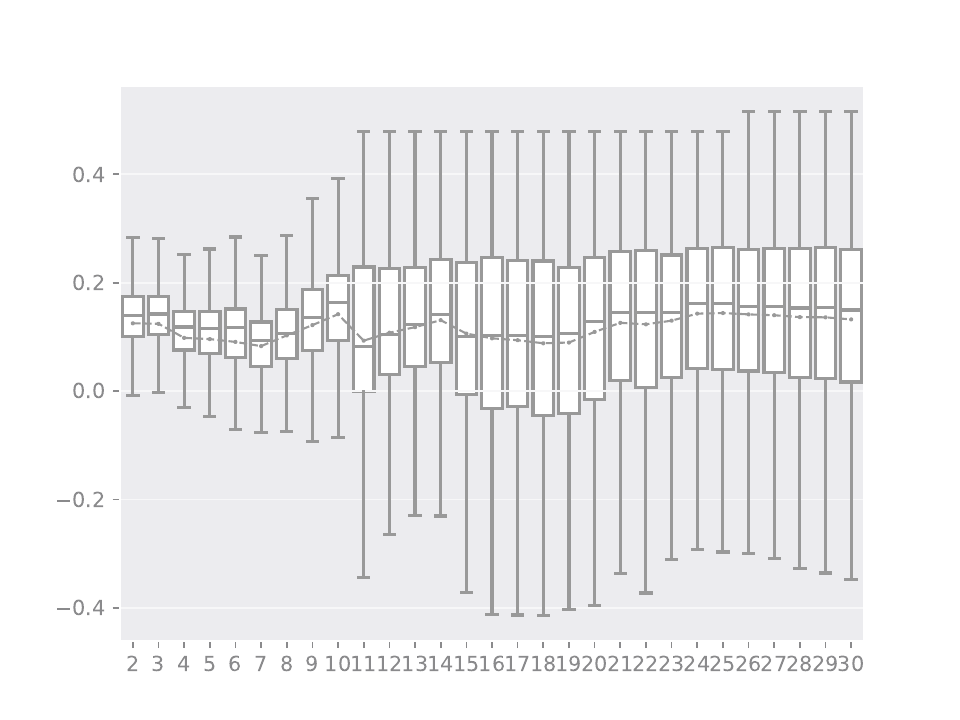}}
\end{minipage}%
\begin{minipage}{.5\linewidth}
\centering
\subfloat[]{\label{subfig:clust_compound_hc}\includegraphics[scale=.5]{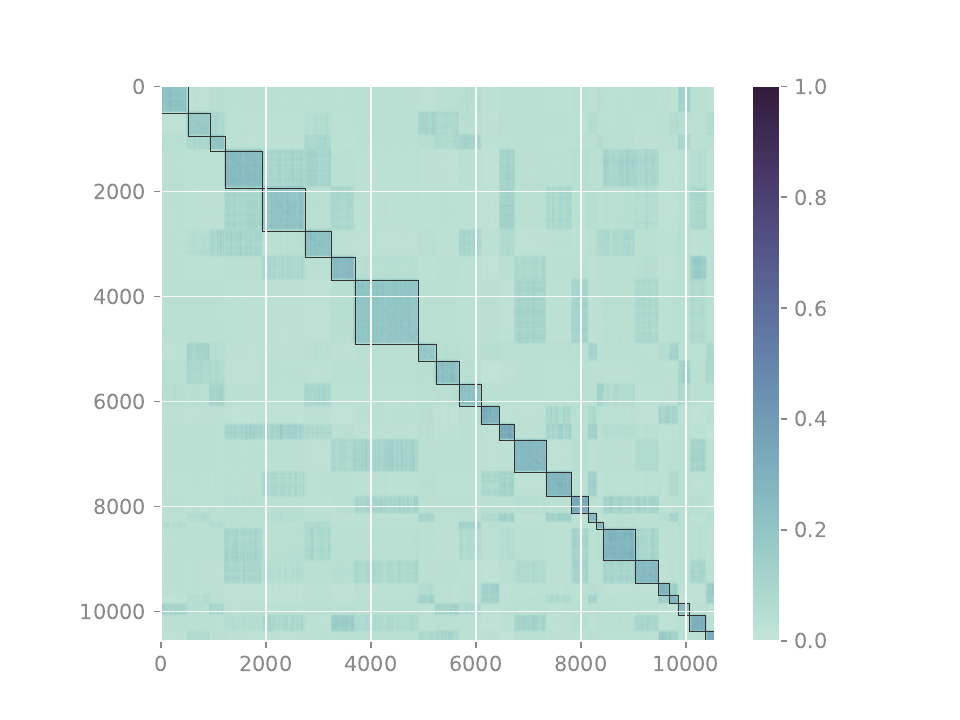}}
\end{minipage}

\caption{Boxplots of the silhouette coefficients for different values of K using the hierarchical clustering algorithm. Thick lines indicate the median, boxes the interquartile range and whiskers the full range of the distribution. The average silhouette is depicted by the dotted line. Partitions of the $\SECO$ similarity matrix with $K = 25$. Squares represent the clusters of variables.}
\label{fig:result_clust_compound_hc}
\end{figure}

\begin{figure}[!htb]

\begin{minipage}{.5\linewidth}
\centering
\subfloat[]{\label{subfig:matseco_thin_silhouette}\includegraphics[scale=.4]{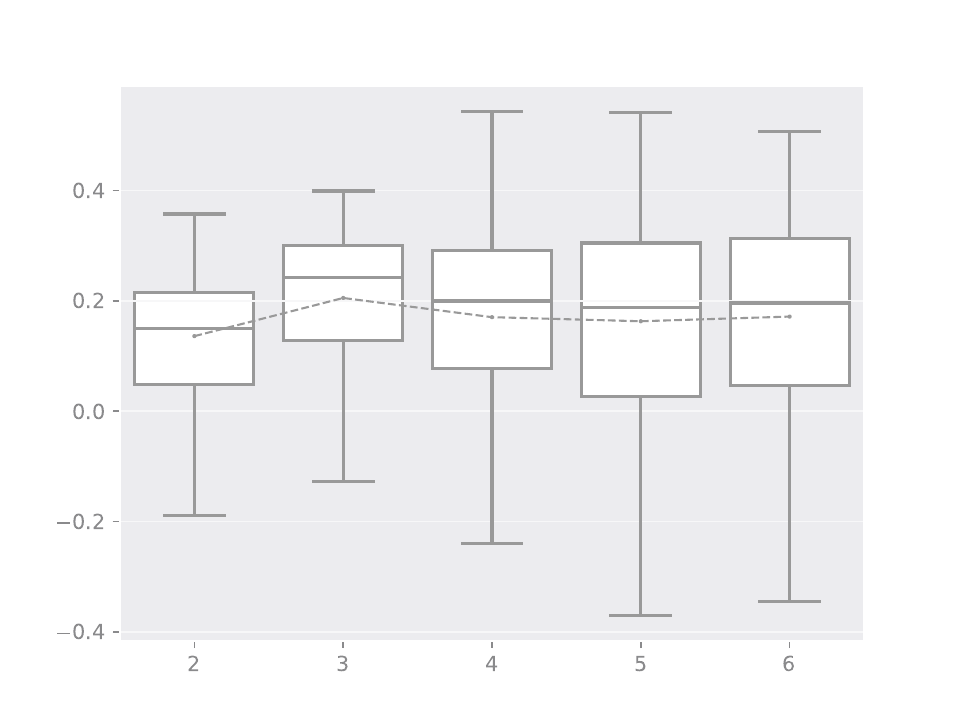}}
\end{minipage}%
\begin{minipage}{.5\linewidth}
\centering
\subfloat[]{\label{subfig:matseco_thin}\includegraphics[scale=.4]{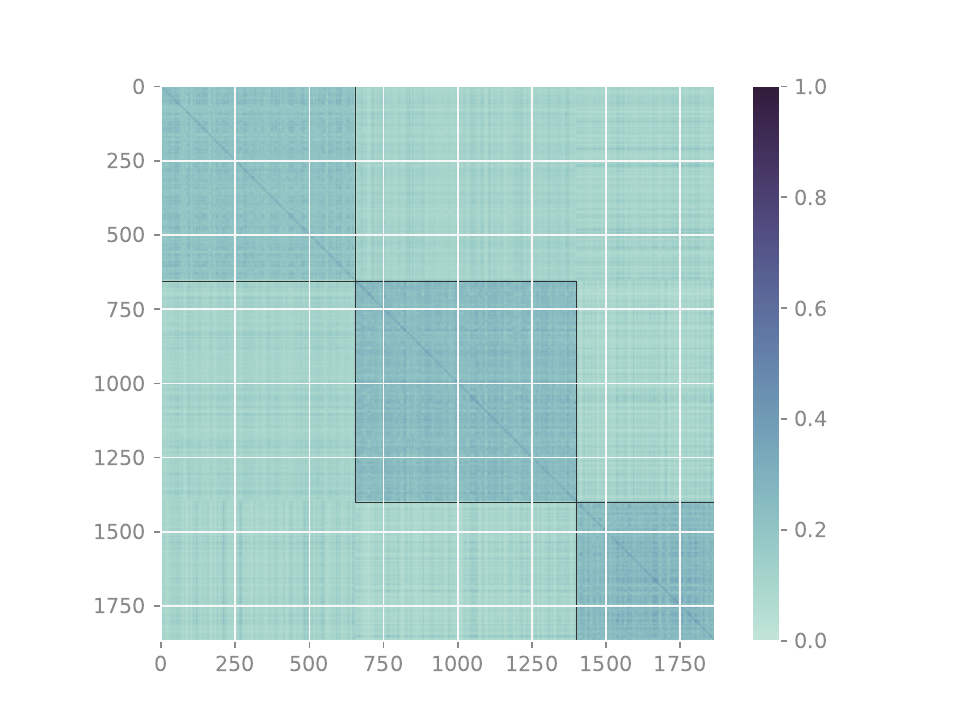}}
\end{minipage}

\caption{Boxplots of the silhouette coefficients for different values of K performed by a hierarchical clustering algorithm on the fourth partition given by Algorithm \ref{alg:rec_pratic}. Thick lines indicate the median, boxes the interquartile range and whiskers the full range of the distribution. The average silhouette is depicted by the dotted line. The clustered matrix SECO is represent in Panel \ref{subfig:matseco_thin} where squares represent the clusters.}
\label{fig:result_clust_thin}
\end{figure}